\documentclass[12pt]{article}
\usepackage{amssymb,amsmath,latexsym,amsthm,amsfonts}
\usepackage[english]{babel}
\usepackage[T1]{fontenc}
\usepackage{color}
\usepackage{graphics}
\usepackage{graphicx}

\usepackage[colorinlistoftodos]{todonotes}

\newtheorem{theorem}{Theorem}[section]

\newtheorem{proposition}[theorem]{Proposition}

\newtheorem{remark}{Remark}[section]

\newtheorem{assumption}{Assumption}

\def\1g{1\hskip -3pt \mbox{l}}

\small\normalsize

    \title{On the correlation analysis of illiquid stocks}

\author{
{\sc
Hamdi Ra\"{\i}ssi\footnote{Instituto de Estad\'{\i}stica, PUCV,
Errazuriz 2734, Valpara\'{\i}so, CHILE; email: hamdi.raissi@pucv.cl. The ANID funding Fondecyt 1201898 is acknowledged.}
}
}

\begin{document}

    \maketitle

\abstract{The serial correlations of illiquid stock's price changes are studied, allowing for unconditional heteroscedasticity and time-varying zero returns probability. Depending on the set up, we investigate how the usual autocorrelations can be accommodated, to deliver an accurate representation of the price changes serial correlations. We shed some light on the properties of the different serial correlations measures, by mean of Monte Carlo experiments. The theoretical arguments are illustrated considering shares from the Chilean stock market.}

\quad


\vspace*{.4cm}\textit{JEL Classification:} C13; C22; C58.

\newpage

\section{Introduction}
\label{intro}

In quantitative finance, it is common to examine the serial correlations of financial assets.
For a daily stock return $r_t$, the correlations $corr(r_t,r_{t-h}$) ($h>0$), can be used for the momentum trading strategy. On the other hand, there is a huge literature considering the serial correlations to test if financial variables are efficient (see e.g. Harris and Kew (2014), K\"{o}chling, M\"{u}ller and Posch (2019), Hurn, Martin, Phillips and Yu (2021,p41) and references therein). The autocorrelations can be investigated to specify the so-called ARMA-GARCH models (see Francq and Zako\"{\i}an (2019) Section 5.2, or McNeil, Frey and Embrechts (2015) Section 4.2.3). Several tools for investigating the serial correlations, taking into account for nonlinearities and/or heteroscedasticity, are available in the literature (see e.g. Romano and Thombs (1996), Lobato, Nankervis and Savin (2002), Patilea and Ra\"{\i}ssi (2013) or Dalla, Giraitis and Phillips (2019)). Such kind of tools are widely used nowadays for many tasks. However, at the best of our knowledge, the presence of unconditional changes for the zeros returns probability is not taken into account in the correlations analysis of returns.

The main message conveyed by this study, is that the non-standard effects of the zero returns scheme, may lead to a misleading estimation of the financial variables serial correlations. First, let us underline that it is advisable to concentrate on the non zero returns. To see this through simple considerations, let us suppose for instance that $corr(r_t,r_{t-1})<0$ ($h=1$), with $E(r_t)=0$. Hence, it can be expected that observing $r_{t-1}<0$ is likely to be followed by $r_{t}>0$, and \emph{vice et versa}. Nevertheless, there is obviously no reason to think that a negative correlation, makes $r_{t-1}=0$ informative with regard to predicting $r_{t}<0$ or $r_{t}>0$. In short, the zero returns are useless, if we are interested in price movements. As a consequence, we aim to investigate how past price changes can be possibly correlated with present price changes, correcting the effect induced by $0<P(r_t=0)<1$ (i.e. use $corr(r_t,r_{t-h}|r_tr_{t-h}\neq0)$ rather than $corr(r_t,r_{t-h}))$. We have

\begin{equation*}\label{correq}
corr(r_t,r_{t-h}|r_tr_{t-h}\neq0)=corr(r_t,r_{t-h})\frac{P(r_t\neq0)}{P(r_tr_{t-h}\neq0)},
\end{equation*}
assuming $E(r_t)=0$ and provided that $P(r_tr_{t-h}\neq0)>0$, $h=0,1,\dots$, which shows that considering the standard $corr(r_t,r_{t-h})$ for our task, can be quite misleading if $0<P(r_t\neq0)<1$. However, let us clarify that we are not claiming that zero returns are not informative at all (see Boudoukh, Richardson and Whitelaw (1994), Section 3.1 and references therein on this subject). For instance, it is well known that $r_{t-1}=0$ could be a piece of information for expecting whether $r_{t}=0$ or not (see Truquet (2019)). Actually, we shall use the possible non constant structure of the zero returns, to improve the estimation of the non zero returns correlations. Finally, let us clarify that $corr(r_t,r_{t-h}|r_tr_{t-h}\neq0)$ is useful for specifying the ARMA part of the $(r_t)$ dynamics, while the $corr(|r_t|^\delta,|r_{t-h}|^\delta|r_tr_{t-h}\neq0)$, for some $\delta>0$, is analyzed for the specifying the stylized financial effects, using for instance some GARCH model. The reader is referred to Patilea and Ra\"{\i}ssi (2021) for more details on the $corr(|r_t|^\delta,|r_{t-h}|^\delta|r_tr_{t-h}\neq0)$ analysis.

We aim to solve our issue following the approach introduced by Dunsmuir and Robinson (1981), Section 2, to take into account for missing values in time series. This consists in considering an imputation process $(a_t)$, such that $a_t=1$ if the series is observed at date $t$, and $a_t=0$ if it is missing. This process, originally introduced by Parzen (1963), is called the "amplitude modulated" sequence in the literature. Then, for some partially observed process $(x_t)$, the empirical serial correlations are computed using $z_t=a_tx_t$ (i.e. replacing the missing observations by zeros). Note that classically $P(a_t=1)$ is assumed constant (see e.g. Stoffer and Toloi (1992) or Bondon and Bahamonde (2012) among others). Such procedure implies to correct adequately the empirical correlations, in order to retrieve a standard normal asymptotic distribution (see Dunsmuir and Robinson (1981)). In our context, as zero returns means irregularly observed price changes, it will be mathematically convenient to express our problem in a similar way. More precisely, we define $a_t=1$ if $r_t\neq0$ at date $t$, and $a_t=0$ otherwise (see Hurn, Martin, Phillips and Yu (2021,p33) for intraday data). This makes our problem close to that of the missing value case, although $r_t$ is not formally missing when it is equal to zero. Nevertheless, the process $(a_t)$ is possibly non-stationary in our case. In addition, the presence of unconditional heteroscedasticity may also induce additional distortions in the autocorrelations measure. For these reasons, a corrected tool will be necessary to depict the correlation structure of the data adequately. Noting that we can reasonably assume that $E(r_t)=0$ in general, then the corrections introduced in the paper are multiplicative. Hence, the main goal of the paper is to solve the adequate \emph{estimation} of $corr(r_t,r_{t-h}|r_tr_{t-h}\neq0)$ for eventual \emph{linear} specification. On the other hand, let us underline that since $E(|r_t|^\delta)\neq0$, then the analysis of $corr(|r_t|^\delta,|r_{t-h}|^\delta|r_tr_{t-h}\neq0)$ lead to correct the \emph{testing} of the presence of the persistence of shocks or some long memory effects in the \emph{volatility}, see Patilea and Ra\"{\i}ssi (2021). Now let us examine the following examples illustrating that non-standard features for $P(a_t=1)$ can be commonly observed for individual stocks.

\subsection{An example from the Chilean stock market}
\label{chi-stock}

It is important to note that stocks with a large amount of zero returns, are quite common in all financial markets. In particular, it is well known that such a feature is especially observed for emerging markets (see e.g. Lesmond (2005,p424-426)). For the sake of illustration of the arguments presented above, we considered the closing prices of some firms taken from the Chilean stock market.\footnote{The author is grateful to Andres Celedon for research assistance.} The data can be downloaded from the Yahoo Finance website. In Figure \ref{ns-stocks}, we plotted the daily log-returns of stocks, that seem to exhibit non-constant zero returns probabilities. For several stocks, we can observe a large amount of years where the probability of observing a price change has a given level, and then, a quick shift can be noticed. 
Examining the Molymet stock in detail, we can observe a quick increase of the probability of observing a daily price change in the end of 2010. This could be explained by a capital increase of more than 216 million US Dollars, announced during the extraordinary shareholders meeting by August 13th, 2010. Note that more than 12.7 million new stocks have been issued in the end of October/early November 2010 on that occasion. Since then, it appears that the probability of a daily price change has never declined to its previous level.
The quick increase of the daily price change probability for the Security bank stock, may be explained by the merger by absorption of the Dresdner Bank Lateinamerika in September 2004. In addition, the Security Bank issued more than 32.8 million new stocks after the capital increase announced during the extraordinary shareholders meeting by December 29th, 2004. Abrupt decreases of the liquidity levels can be also observed. For instance, the quick decrease of the daily trade probability for the Provida stock seems to be a consequence of the take-over bid of Metlife on Provida during September 2013. Other stocks seem to display long-run trending behaviors with regard to the non zero returns probability, as the Paz, HF or Cruzados stocks. The increase of $P(a_t=1)$ for the Vapores stock during the 2000's can be observed in a similar way for many stocks of the Chilean Stock market. Such a behavior is certainly the consequence of the fast development of the Chilean economy during this period.
The consequences of such non-standard behaviors on the serial correlation estimation are further investigated in Section \ref{real-analysis}.

The structure of this study is as follows. In section \ref{non-stationary-t-f}, the theoretical framework is outlined. The main setting for the analysis of  $corr(r_t,r_{t-h}|r_tr_{t-h}\neq0)$ is also defined. In Section \ref{first-order}, the relevant serial correlation we may need to estimate the correlation structure of illiquid stocks is first highlighted. Subsequently, corrections for the autocorrelations in the stationary and non stationary cases are presented in Section \ref{stationary-weak} and Section \ref{non-stationary}.
Thereafter, Monte Carlo experiments are conducted to assess the finite sample properties of the serial correlations estimators in Section \ref{num-exp}. We also make additional comments on the autocorrelations of the Chilean stock market data presented above. Some concluding remarks are given in Section \ref{concl}.

Throughout the paper the following notation is used. Let us consider a one-period profit-and-loss random variable $r_t$. We assume that $r_1,\dots,r_n$ are observed, with $n$ the sample size. We denote by $\mathcal{F}_{t-1}=\sigma\{r_{l}:l<t\}$ the $\sigma$-field generated by the past values of $r_t$. Recall that $a_t=0$ if $r_t=0$ and $a_t=1$ if not. The almost sure convergence is denoted by $\stackrel{a.s.}{\longrightarrow}$, and the convergence in distribution is signified by $\stackrel{d}{\longrightarrow}$. The convergence in probability is denoted by $\stackrel{p}{\longrightarrow}$.

\section{The theoretical framework}
\label{non-stationary-t-f}

In time series econometrics, there is a significant body of works dealing with long-run phenomena. For example, reference can be made to the huge cointegration literature. In the field of financial returns, Mikosch and St\u{a}ric\u{a} (2004) and St\u{a}ric\u{a} and Granger (2005), have documented the presence of a long-run unconditional time-varying variance. This pattern can be noted especially when the period under study is relatively large. In order to take into account for such a behavior, Engle and Rangel (2008) have proposed the spline GARCH model. This model combines stochastic effects in order to capture the short run GARCH effects, and a deterministic specification to describe the long-run unconditional changes of the variance. The reader is also referred to Hafner and Linton (2010) for the multivariate case. In this paper, we rely on such an approach to describe the returns non-stationary behavior:

\begin{equation}\label{eq-r-t}
r_t=\sigma_t\epsilon_t,
\end{equation}
where $(\epsilon_t)$ is a random process specified below, and $\sigma_t$ is non-constant deterministic, fulfilling the following assumption.

\begin{assumption}[Unconditional time-varying variance]\label{tv-var}
The $\sigma_{t}$'s are given by $\sigma_{t}=v(t/n)$,
    where $v(\cdot)$ is a
    measurable deter\-ministic and strictly positive function on the interval $(0,1]$, such that
    $\sup_{r\in(0,1]}|v(r)|<\infty$, and satisfies a piecewise Lipschitz
    condition on  $(0,1]$.\footnote{The piecewise Lipschitz condition means: there exists a positive integer $p$ and some mutually disjoint intervals $I_1,\dots,I_p$ with $I_1\cup\dots\cup I_p=(0,1]$ such that $v(r)=\sum_{l=1}^p v_l(r){\bf 1}_{\{r\in I_l\}},$ $r\in(0,1],$ where $v_1(\cdot),\dots,v_p(\cdot)$ are Lipschitz smooth functions on $I_1,\dots,I_p,$ respectively.}
\end{assumption}

Let us underline that the Dahlhaus (1997) rescaling device is often used in order to describe long-run effects (see e.g. Phillips and Xu (2006), Xu and Phillips (2008), Cavaliere  and Taylor (2007,2008), Patilea and Ra\"{\i}ssi (2012,2013) or Wang, Zhao and Li (2019)). A triangular notation should be adopted in view of Assumption \ref{tv-var}. However, the double subscript is suppressed for notational simplicity. In view of the above, $V(r_t)$ can be non-constant, but we also possibly allow $V(r_t|\mathcal{F}_{t-1})$ to be stochastic, as exposed below. Let us mention that Patilea and Ra\"{\i}ssi (2014) have proposed tests, for determining if a purely deterministic specification of the variance is sufficient, or additional GARCH effects must be taken into account.

In a similar way to the above literature, we suppose that the unconditional probability of observing a daily price change is time-varying. More precisely, it is assumed that the probabilities $P(a_ta_{t-h}=1)$, $h=0,1,\dots,m$, are equal to deterministic non constant functions $g_h(t/n)$. Note that we do not exclude the presence of stochastic short time effects (i.e. the $P(a_t=1|\mathcal{F}_{t-1})$ is stochastic), so $(a_t)$ is possibly dependent.

\begin{assumption}[Non constant unconditional probabilities]\label{tv-prob} The time-varying probabilities $P(a_ta_{t-h}=1)$ are given by $g_h(t/n)$,
    where $g_h(\cdot)$, $h=0,1,\dots,m$, are
    measurable deter\-ministic functions, such that $0<c< g_h(t/n)\leq1$ on the interval $(0,1]$, and each $g_h(\cdot)$ satisfies a piecewise Lipschitz condition on $(0,1]$.\footnote{For $u<0$, the different functions considered in Assumption \ref{tv-var} and \ref{tv-prob} are set constant, that is $g_h(u)=\lim_{u\downarrow0}g_h(u)$ and $v(u)=\lim_{u\downarrow0}v(u)$.}
\end{assumption}

Assumption \ref{tv-var} and \ref{tv-prob} are quite general, and allow for a wide type of specifications, as trending behaviors or abrupt breaks for instance. If the functions $g_h(\cdot)$ or $v(\cdot)$ are non-constant, then the returns process is not stationary. In view of Assumption \ref{tv-prob}, we have $P(a_ta_{t-h}=1)<1$, at least for some subperiod in the sample. Concerning the non-constant variance, the studies mentioned above provide evidence of unconditional heteroscedasticity for financial variables with $P(a_t=1)=1$. Then, it is likely that the same holds when $P(a_t=1)<1$. In Figure \ref{ns-stocks}, we can observe abrupt breaks as well as trending behaviors, for the probabilities of observing a daily price change. It is reasonable to think that such structural changes may happen for both the non zero returns probability and the variance. 
In addition, let us mention that, using the above mentioned rescaling device, Truquet (2019) studied the pointwise estimation of the time-varying transition probabilities

\begin{equation}\label{cond-prob}
P(a_t=1|a_{t-1}=1)=\frac{P(a_ta_{t-1}=1)}{P(a_{t-1}=1)},
\end{equation}
with application to capture non-constant features for South African stock returns. In our case, (\ref{cond-prob}) suggests that $g_h(\cdot)$, $h=0,1,\dots,m$, are non constant. Finally, let us make the following comment on Assumption \ref{tv-var} and \ref{tv-prob} jointly.

\begin{remark}\label{on-label}
If $P(a_t=1)$ is time-varying, then it is easy to see that $(r_t)$ will be heteroscedastic whether $\sigma_t$ is constant or not. However, in the sequel we will refer to the case where $\sigma_t$ is time-varying (resp. constant), as heteroscedastic (resp. homoscedastic) in the sense that the prices changes (i.e. eliminating the zero returns) are heteroscedastic (resp. homoscedastic). When $\sigma_t$ is constant and $P(a_t=1)$ is time-varying, we will simply say that $(r_t)$ is not stationary, or have a time-varying zero returns probability to avoid confusion, although it is heteroscedastic. In short, $V(r_t|a_t=1)$ is constant if and only if $\sigma_t$ is constant, in view of Assumption \ref{identifiability} below.
\end{remark}

The identifiability of (\ref{eq-r-t}) is ensured by considering the following assumption.

\begin{assumption}[Identifiability]\label{identifiability}
We assume that $E(\epsilon_t|a_t=1)=0$ and $E(\epsilon_t^2|a_t=1)=1$.
\end{assumption}

Note that the above condition is expressed conditional to $a_t=1$, as we concentrate on the price changes. Usually in models such as (\ref{eq-r-t}) we set $E(\epsilon_t^2)=1$. However, allowing for a possible time-varying $P(a_{t}=1)$, then $E(\epsilon_t^2)$ is in general non-constant in our case. The condition $E(\epsilon_t|a_t=1)=0$ implies that $E(r_t)=0$, which is realistic since the stocks returns means are not significant in general.

If the sources of non stationarity are removed in (\ref{eq-r-t}), that is the unconditional heteroscedasticity and the time-varying probability, it can be expected that the remainder is stationary. The following assumption formalizes this idea.

\begin{assumption}[Strict stationarity in the time-varying framework]\label{stationarity}
Given $a_t\times\dots\times a_{t+k}=a_{t+h}\times\dots\times a_{t+k+h}=1$, the vectors $(\epsilon_t,\dots,\epsilon_{t+k})'$ and $(\epsilon_{t+h},\dots,\epsilon_{t+k+h})'$ have the same joint distribution for any $k\in\mathbb{N}$ and any $(t,h)'\in\mathbb{Z}^2$.
\end{assumption}


%
%
%

Now we turn to the dependence set-up used for $(r_t)$. In the literature, a wide variety of models are available to capture the financial assets dynamics. Among many others, reference can be made to the large class of GARCH models first introduced by Engle (1982) (see the Bollerslev (2008) glossary). 
Hence, we consider a weak dependence concept, that allow for a wide range of specifications. In the sequel, we denote by
$$\alpha_{\epsilon}(h)= \sup_{A\in \sigma(\epsilon_u, u\leq t), B\in \sigma(\epsilon_u, u\geq
t+h)}\left|P(A\cap B)-P(A)P(B)\right|,$$
the mixing coefficients of $(\epsilon_t)$. The reader is referred
to Davidson (1994) for details about the $\alpha$-mixing coefficients. Let us also introduce
$\|\epsilon_t\|_q=\left(E\|\epsilon_t\|^q\right)^{1/q}$, where $\|.\|$ denotes the
Euclidean norm.\\

\begin{assumption}[Weak dependence condition for strong consistency]\label{mixing-sc}
There exist $\mu_1=2+2\nu$, $\nu>0$, 
and $1\leq\mu_2<1+\nu<2\mu_2$,
such that $\sum_{h=0}^{\infty}\{\alpha_{\epsilon}(h)\}^{1-1/2\mu_2}$ $<$ $\infty$ and $\sup_t\left\|\epsilon_t\right\|_{\mu_1}$ $<$ $\infty$.
\end{assumption}

\begin{assumption}[Martingale difference condition for asymptotic normality]\label{mixing-an}
(a) We assume that $E(\epsilon_t|\mathcal{F}_{t-1})=0$, and $sup_t\left\|\epsilon_t\right\|_{4}$ $<$ $\infty$.
(b) There exists $\tilde{\mu}_1=4+4\nu$ such that $\sup_t\left\|\epsilon_t\right\|_{\tilde{\mu}_1}$ $<$ $\infty$.
\end{assumption}

Note that $(r_t)$ is allowed to be correlated under our mixing assumption. In Assumption \ref{mixing-sc}, we make the usual trade-off between moment and mixing conditions (see Phillips (1987) and references therein).
The Assumption \ref{mixing-an}(a) is consistent with an efficient market null hypothesis. Other dependence, prediction assumptions can be used to state the convergence and the asymptotic normality results below, as soon as convergence results are available. For instance, the reader is referred to Davidson (1994), Chapters 16 and 17, for the mixingale (asymptotic unpredictability) or the Near Epoch Dependence (NED) concepts. 

\section{Serial correlations in presence of zero returns}
\label{first-order}

In this part, the serial correlations analysis of illiquid stocks is discussed in the light of the assumptions made above.
For mathematical ease, the following notations are adopted. The correlations of the process $(r_t)$ are given by

$$\rho_{0t}(h):=corr(r_t,r_{t-h})=\gamma_{0t}(h)\gamma_{0t}(0)^{-1},\:\mbox{with}\:\gamma_{0t}(h):=E(r_tr_{t-h}).$$
As usual, we consider vectors of these autocorrelations:
$\Gamma_{0t}(m):=(\rho_{0t}(1),\dots,\rho_{0t}(m))'$, for a given $m$, fixed by the practitioner. The subscript zero means that the zero returns are included in the autocorrelations. Under Assumption \ref{tv-var}-\ref{stationarity}, some basic computations show that

\begin{equation}\label{comput}
\rho_{0t}(h)=\frac{\rho_\epsilon(h)g_h(t/n)}{\left\{g_0(t/n)g_0\left((t-h)/n\right)\right\}^{\frac{1}{2}}},
\quad\mbox{where}\quad\rho_\epsilon(h)=E(\epsilon_t\epsilon_{t-h}|a_ta_{t-h}=1),
\end{equation}
is constant. 
In the sequel, we consider $\Gamma_{\epsilon}(m):=(\rho_\epsilon(1),\dots,\rho_\epsilon(m))'$. It is clear that the correlations of $(r_t)$ are not constant, due to the time-varying probabilities of observing a price change. As a consequence, these non-constant linear dynamics may be viewed as spurious. 
On the other hand, the unconditional heteroscedasticity does not play any role in the theoretical relation (\ref{comput}). 
Now, if in addition $P(a_ta_{t-h}=1)$ and $P(a_t=1)$ are assumed constant, then $\rho_{0t}(h)$ does not depend on time:

\begin{equation}\label{auto-stat-case}
\rho_0(h)=\rho_\epsilon(h)\frac{P(a_ta_{t-h}=1)}{P(a_t=1)}.
\end{equation}
In such a case, the notations can be simplified into $\gamma_{0t}(h)=\gamma_{0}(h)=E(r_tr_{t-h})$, $\rho_{0}(h):=\gamma_{0}(h)/\gamma_{0}(0)$, and $\Gamma_{0}(m):=(\rho_{0}(1),\dots,\rho_{0}(m))'$. It is interesting to note that the dependence effects in $(a_t)$ should be carefully taken into account to correctly measure the correlation structure. 
%
This last case is similar to the classical treatment of missing values in time series (see e.g. Stoffer and Toloi (1992)). 

In view of Figure \ref{ns-stocks}, non-stationary effects can be often assumed for $(a_t)$. 
In addition, for computational reasons some bias could be introduced in the estimation step from the non-constant time-varying variance. These combined non-stationary behaviors can lead to erroneous conclusions. 
For instance, suppose that we are interested to measure the eventual non-zero correlations of a stock at different periods. In the case of the Molymet stock, using (\ref{comput}) may spuriously lead to conclude that the potential linear link between the returns is affected by the 2010 increase of capital. 
We can conclude therefore that if we are interested in analyzing the correlation structure of the price changes, then $\rho_\epsilon(h)$ should be considered rather than $\rho_{0t}(h)$.
The methodology to handle the non-standard behaviors described above, will consist in applying a scale correction to the classical autocorrelation estimator

\begin{equation*}\label{corr-classic}
\widehat{\Gamma}_0(m):=(\hat{\rho}_0(1),\dots,\hat{\rho}_0(m))'\:\mbox{with}\:
\hat{\rho}_0(h):=\hat{\gamma}_0(h)\hat{\gamma}_0(0)^{-1},\:\mbox{and}\:\hat{\gamma}_0(h):=\frac{1}{n}\sum_{t=1+h}^{n}r_t r_{t-h},
\end{equation*}
in order to obtain a consistent estimator of $\Gamma_{\epsilon}(m):=(\rho_\epsilon(1),\dots,\rho_\epsilon(m))'$. 

\subsection{The stationary case}
\label{stationary-weak}

In this part, we suppose that $\sigma_t=\sigma$, $P(a_t=1)$ and $P(a_ta_{t-h}=1)$ are constant. 
%
%
In view of (\ref{auto-stat-case}), let us introduce the autocorrelations estimator

\begin{equation}\label{co-ds}
\hat{\rho}_{pr}(h):=\frac{\hat{\gamma}_0(h)\hat{\gamma}_a(h)^{-1}}{\hat{\gamma}_0(0)\hat{\gamma}_a(0)^{-1}}
=\hat{\rho}_0(h)\frac{\hat{\gamma}_a(0)}{\hat{\gamma}_a(h)},\:\mbox{where}\:\hat{\gamma}_a(h):=n^{-1}\sum_{t=h+1}^{n}a_ta_{t-h},
\end{equation}
and $$\widehat{\Gamma}_{pr}(m):=\widehat{\Lambda}_{\gamma}\widehat{\Gamma}_0(m),\quad
\widehat{\Lambda}_{\gamma}=diag\left(\frac{\hat{\gamma}_a(0)}{\hat{\gamma}_a(1)},\dots,\frac{\hat{\gamma}_a(0)}{\hat{\gamma}_a(m)}\right),
$$
see e.g. Stoffer and Toloi (1992) for this kind of correction. The subscript "pr" means that we remain with corrections based on probabilities. Introduce

$$\Lambda_{\gamma}=diag\left(\frac{P(a_t=1)}{P(a_ta_{t-1}=1)},\dots,\frac{P(a_t=1)}{P(a_ta_{t-m}=1)}\right).$$
The following proposition ensures that $\widehat{\Gamma}_{pr}(m)$ delivers a correct picture of the price changes autocorrelations. Recall that $\Gamma_{\epsilon}(m)=(\rho_\epsilon(1),\dots,\rho_\epsilon(m))'$.

\begin{proposition}\label{propostu2}
Suppose that $(r_t)$ is strictly stationary ergodic, such that $E(r_t^2)<\infty$ and $E(r_t)=0$. Then, we have $\widehat{\Gamma}_{pr}(m)\stackrel{a.s.}{\longrightarrow}\Gamma_{\epsilon}(m),$ as $n\to\infty$.
\end{proposition}

Before moving forward to the case where $(a_t)$ is dependent, let us mention that the autocorrelations are often displayed with the confidence bounds under the hypothesis $\Gamma_{0t}(m)=\Gamma_{\epsilon}(m)=0$. For instance, it may serve to identify an horizon $h$ where the autocorrelations seem to vanish. 

\begin{proposition}\label{propostu2bis}
Suppose that the assumptions of Proposition \ref{propostu2} hold true. Then, under Assumption \ref{mixing-an}(a), we have
\begin{equation}\label{eq-norm}
n^{\frac{1}{2}}\widehat{\Gamma}_{pr}(m)\stackrel{d}{\rightarrow}
     \mathcal{N}(0,\Lambda_{\gamma}\Sigma_{pr}\Lambda_{\gamma}),
\end{equation}
as $n\to\infty$, and where $\Sigma_{pr}=\sigma_r^{-4}\Sigma_{\Upsilon}$, $\sigma_r^{2}=E(r_t^2)$
with $\Sigma_{\Upsilon}=E\left(\Upsilon_t\Upsilon_{t}'\right)$, and  $\Upsilon_t=(r_tr_{t-1},\dots,r_tr_{t-m})'$.
In addition, we have as $n\to\infty$,
\begin{equation}\label{eq-sig}
\widehat{\Sigma}_{\Upsilon}:=n^{-1}\sum_{t=1+m}^{n}\Upsilon_t\Upsilon_{t}'\stackrel{a.s.}{\longrightarrow}E\left(\Upsilon_t\Upsilon_{t}'\right),
\quad\hat{\sigma}_r^2:=n^{-1}\sum_{t=1}^{n}r_t^2\stackrel{a.s.}{\longrightarrow}\sigma_r^2,
\end{equation}
and
\begin{equation}\label{eq-a}
\hat{\gamma}_a(h)\stackrel{a.s.}{\longrightarrow}P(a_ta_{t-h}=1),\:h=0,1,\dots,m.
\end{equation}
\end{proposition}

The asymptotic covariance matrix in (\ref{eq-norm}) can be estimated using (\ref{eq-sig}) and (\ref{eq-a}). 


\subsection{The non stationary case}
\label{non-stationary}

Now both unconditional heteroscedasticity and time-varying zero returns probability are allowed. The cases where we have only unconditional heteroscedasticity, or only non-constant zero returns probability are briefly discussed in Remarks \ref{het-solo} and \ref{prob-solo} below. 
From (\ref{comput}), the unconditional heteroscedasticity does not play any role in $\Gamma_{0t}(m)$. However, in practice the numerator and denominator are computed separately on the whole sample. As a consequence, some heteroscedasticity effects are produced that must be taken into account for the estimation of $\Gamma_{\epsilon}(m)$. To this end, let us introduce the following estimator:

\begin{equation}\label{co-ns}
\widehat{\rho}_{vpr}(h):=\frac{\hat{\gamma}_{0}(0)}{\hat{\gamma}_{ar^2}(h)}\widehat{\rho}_0(h),
\quad\mbox{and}\quad
\widehat{\Gamma}_{vpr}(h):=\widehat{\Phi}_\gamma\widehat{\Gamma}_0(h),
\end{equation}
where

$$\widehat{\Phi}_\gamma=diag\left(\frac{\hat{\gamma}_{0}(0)}{\hat{\gamma}_{ar^2}(1)},
\dots,\frac{\hat{\gamma}_{0}(0)}{\hat{\gamma}_{ar^2}(m)}\right),
\:\mbox{and}\:\hat{\gamma}_{ar^2}(h)=(n-h)^{-1}\sum_{t=1+h}^{n}r_t^2\frac{\hat{p}_{t,t-h}}{\hat{p}_{t}}.$$
The kernel estimators of $P(a_ta_{t-h}=1)$ and $P(a_t=1)$ are given by\footnote{If $\hat{p}_{t}=0$, then the corresponding term in the sum defining $\hat{\gamma}_{ar^2}(h)$ is set equal to zero. Note that in such a case, $r_t^2$ is often equal to zero.}

$$\hat{p}_{t,t-h}=\sum_{j=1+h}^{n} w_{tj}(b)a_ja_{j-h},\quad
\mbox{and}\quad\hat{p}_{t}=\sum_{j=1}^{n} w_{tj}(b) a_j,$$
where $w_{tj}=\left(\sum_{j=1}^nK_{tj}\right)^{-1}K_{tj}$, and
    $$K_{ti}=\left\{
                  \begin{array}{c}
                    K((t-i)/nb)\quad \mbox{if}\quad t\neq i\\
                    0  \quad\mbox{if}\quad t=i,\\
                  \end{array}
                \right.$$
The kernel function on the real line $K(\cdot)$, and the bandwidth $b$ fulfill the following conditions.

\begin{assumption}[Kernel and bandwidth]\label{k-b}
\begin{itemize}
\item[(a)] $K(\cdot)$ is a continuous kernel function defined on the real line with compact support, such that $0\leq \sup_z K(z)<R$ for some finite real number $R$, and $\int_{-\infty}^{\infty}K(z)dz=1$.
\item[(b)] As $n\to\infty$, $b+\frac{1}{nb^2}\to0$.
\end{itemize}
\end{assumption}

Of course different bandwidths can be used for $\hat{p}_{t,t-h}$ and $\hat{p}_{t}$. Nevertheless, for ease of exposition, the kernel estimators are presented with the same bandwidth. Define
$$\Phi_\gamma=diag\left(\frac{\gamma_{0}(0)}{\gamma_{ar^2}(1)},
\dots,\frac{\gamma_{0}(0)}{\gamma_{ar^2}(m)}\right),$$
where
$$\gamma_{0}(0)=\int_{0}^{1}v^2(s)g_0(s)ds,\quad\mbox{and}\quad\gamma_{ar^2}(h)=\int_{0}^{1}v^2(s)g_h(s)ds.$$
We are now ready to state the following asymptotic results, in our non standard framework.

\begin{proposition}\label{propostu3}
Suppose that the sequence $(r_t)$ fulfills (\ref{eq-r-t}). Under Assumption \ref{tv-var}-\ref{mixing-sc}, then as $n\to\infty$, $\widehat{\Gamma}_{vpr}(m)\stackrel{p}{\longrightarrow}\Gamma_{\epsilon}(m).$
\end{proposition}

\begin{proposition}\label{propostu3bis}
Suppose that Assumption \ref{mixing-an} holds true. Then, under the assumptions of Proposition \ref{propostu3}, we have
$$n^{\frac{1}{2}}\widehat{\Gamma}_{vpr}(m)\stackrel{d}{\rightarrow}
     \mathcal{N}(0,\Phi_{\gamma}\Sigma_{vpr}\Phi_{\gamma}),$$
as $n\to\infty$, and where $\Sigma_{vpr}=\gamma_{0}(0)^{-2}\widetilde{\Sigma}_{\Upsilon}$, with $\widetilde{\Sigma}_{\Upsilon}:=\lim_{n\to\infty}n^{-1}\sum_{t=1}^n
E\left(\Upsilon_t\Upsilon_{t}'\right)$. In addition, as $n\to\infty$, we have $\widehat{\widetilde{\Sigma}}_{\Upsilon}:=n^{-1}\sum_{t=1+m}^{n}\Upsilon_t\Upsilon_{t}'
\stackrel{a.s.}{\longrightarrow}\widetilde{\Sigma}_{\Upsilon}.$
\end{proposition}

Now we shall compare the stationary case and the non-stationary case. If we suppose that the conditions of Proposition \ref{propostu2} hold (the variance and the zero return probability are constant), then we can write:

$$\int_{0}^{1}v^2(s)g_h(s)ds=\sigma^2P(a_ta_{t-h}=1),\:\mbox{and}\:
\int_{0}^{1}v^2(s)g_0(s)ds=\sigma^2P(a_t=1).$$
As a consequence, the correction (\ref{co-ns}) can be applied to the stationary case as well. The reverse is not true in general, as the combined effects of the heteroscedasticity and non-constant liquidity levels must be specifically taken into account. On the other hand, the "vpr" needs to estimate the probability structure, on the contrary to the "pr" correction. In the real data analysis below, a tool for identifying the "pr" and "vpr" cases is provided.

Let us end this section, by considering briefly the case where we have unconditional heteroscedasticity, and constant probability on one hand; and the case where the time-varying probability is the unique source of non-stationarity. 

\begin{remark}[Unconditional heteroscedasticity, constant zero returns probabilities]\label{het-solo}
Let us consider the assumptions of Proposition \ref{propostu3}, strengthened by the assumption that $(a_t)$ is strictly stationary (i.e. $P(a_ta_{t-h}=1)$, $h=0,1,\dots,m$, are constant). In such a case, we have

$$\int_{0}^{1}v^2(s)g_0(s)ds=P(a_t=1)\int_{0}^{1}v^2(s)ds,$$
and

$$\int_{0}^{1}v^2(s)g_h(s)ds=P(a_ta_{t-h}=1)\int_{0}^{1}v^2(s)ds,$$
for any $h\in\{1,\dots,m\}$. Therefore, the non-constant variance effect vanishes, as we can write $\Phi_\gamma=\Lambda_{\gamma}$. Hence, in view of (\ref{maipu}), it is easy to see that
$$\widehat{\Gamma}_{pr}(m)\stackrel{a.s.}{\longrightarrow}\Gamma_{\epsilon}(m),$$
in the case of a time-varying variance, and a constant zero returns probability. 
\end{remark}

\begin{remark}[Constant variance, time-varying zero returns probabilities]\label{prob-solo}
Now, we assume that the conditions of Proposition \ref{propostu3} hold true, strengthened by the assumption that $\sigma_t=\sigma$ is constant. In a similar way as above we write

$$\int_{0}^{1}v^2(s)g_0(s)ds=\sigma^2\int_{0}^{1}g_0(s)ds,$$
and
$$\int_{0}^{1}v^2(s)g_h(s)ds=\sigma^2\int_{0}^{1}g_h(s)ds,$$
for any $h\in\{1,\dots,m\}$. Again, we readily have
$$\widehat{\Gamma}_{pr}(m)\stackrel{a.s.}{\longrightarrow}\Gamma_{\epsilon}(m),$$
in the case of constant variance, and time varying zero returns probability.
\end{remark}

\section{An index for deciding between $\widehat{\Gamma}_{pr}(m)$ and $\widehat{\Gamma}_{vpr}(m)$}
\label{var-prob}

In this part, we provide a tool to determine whether $\widehat{\Gamma}_{vpr}(m)$ or $\widehat{\Gamma}_{pr}(m)$ should be used. 
From the above section, $\widehat{\Gamma}_{pr}(m)$ is suitable if we have $\Delta=0$, with

\begin{equation}\label{suitable-ds}
\Delta=\sum_{h=1}^{m}\Delta_h^2,\quad\mbox{and}\quad\Delta_h=\frac{\int_{0}^{1}v^2(s)g_0(s)ds}{\int_{0}^{1}v^2(s)g_h(s)ds}-
\frac{\int_{0}^{1}g_0(s)ds}{\int_{0}^{1}g_h(s)ds}.
\end{equation}
The use of $\widehat{\Gamma}_{vpr}(m)$, corresponds to the situation where there exists at least some $h\in\{1,\dots,m\}$, such that $\Delta_h\neq0$.
The above alternatives can be viewed, in some sense, as a diagnostic for $(r_t)$. If we decide that $\Delta\neq0$, then this is equivalent to assert that the following statements are not true:

\begin{itemize}
\item[(a)] $\sigma_t=\sigma$ and $P(a_ta_{t-h}=1)$ are constant for all $h\in\{1,\dots,m\}$,
\item[(b)] $\sigma_t$ is non-constant, but $(r_t)$ has a constant zero returns probability,
\item[(c)] $\sigma_t=\sigma$, but $(r_t)$ has a time-varying zero returns probability.
\end{itemize}
If decide that all the $\Delta_h$'s are equal to zero, then we can conjecture that (a), (b) or (c) holds true. This last assertion can only be conjectured, as we may have $\Delta=0$ for some specific non-constant $v(\cdot)$ and $g_h(\cdot)$ functions, $h\in\{0,1,\dots,m\}$.

Now, let us introduce the index $\hat{\kappa}_m=\widehat{\Delta}$, where

$$\widehat{\Delta}=\sum_{h=1}^{m}\widehat{\Delta}_h^2,$$
and
$$\widehat{\Delta}_h=\frac{n^{-1}\sum_{t=1}^{n}r_t^2}{(n-h)^{-1}\sum_{t=1+h}^{n}r_t^2\frac{\hat{p}_{t,t-h}}{\hat{p}_{t}}}
-\frac{n^{-1}\sum_{t=1}^{n}a_t}{(n-h)^{-1}\sum_{t=1+h}^{n}a_ta_{t-h}}.$$
The existence of the $\widehat{\Delta}_h$'s is ensured, at least asymptotically, from Assumption \ref{tv-var} and \ref{tv-prob}. 
%

\begin{proposition}\label{index-prop}
Suppose that Assumptions \ref{tv-var}-\ref{mixing-sc} hold true. If $\Delta=0$, then as $n\to\infty$,
$\hat{\kappa}_m\stackrel{p}{\longrightarrow}0.$
\end{proposition}

\begin{proposition}\label{index-prop-power}
Under Assumption \ref{tv-var}-\ref{mixing-sc}, and if there exists $h\in\{1,\dots,m\}$ such that $\Delta_h\neq0$, then we have $\kappa_m\stackrel{p}{\longrightarrow}C>0.$
\end{proposition}

In view of Proposition \ref{index-prop} and \ref{index-prop-power}, if $\kappa_m$ is somewhat far from zero, then it is advisable to use $\widehat{\Gamma}_{vpr}(m)$.

\section{Numerical illustrations}
\label{num-exp}

In this part, the different tools introduced in the previous sections, are first investigated by mean of Monte Carlo experiments. Then, the autocorrelation structures of the above presented Chilean market's stocks are analyzed more deeply.

\subsection{Monte Carlo experiments}
\label{MC}

The pointwise estimation of the correlation $\Gamma_{\epsilon}(1)\neq0$ is studied ($m=1$). Some comments are also made on the index $\kappa_m$ for $m=1$. 
In all the experiments, the following Data Generating Process (DGP) is considered

\begin{eqnarray}\label{dgp}
r_t&=&\tilde{\sigma}_t\tilde{\epsilon}_t,\quad\tilde{\sigma}_t=\tilde{\sigma}(t/n),\\
\tilde{\sigma}(s)&=&\delta_11_{\{(0,0.4]\}}(s)+\left[\left(\frac{(\delta_2-\delta_1)}{0.2}\right)s+3\delta_1-2\delta_2\right]\times1_{\{(0.4,0.6]\}}(s)\nonumber\\
&+&\delta_2\times1_{\{(0.6,1]\}}(s), s\in(0,1],\nonumber
\end{eqnarray}
where $\delta_1$ and $\delta_2$ are parameters given below. Note that $\tilde{\sigma}_t$ and $\tilde{\epsilon}_t$ are only rescales of $\sigma_t$ and $\epsilon_t$ according to Assumption \ref{identifiability}. We set $\tilde{\epsilon}_t=a_{1t}y_t$. In our simulation experiments, the sequence $y_1,\dots,y_n$ is defined in two ways. Firstly, to illustrate possible bias of correlated non-zero returns,  we consider

\begin{equation}\label{ar-1}
y_t=0.3u_{t-1}+u_t,
\end{equation}
if $u_t$ and $u_{t-1}$ are both different from zero, and $y_t=0$ if not. The innovations process $(u_t)$ is not correlated but dependent and defined as follows. Let us consider the process given in Romano and Thombs (1996): $x_t=\Pi_{i=0}^3z_{t-i}$, where $(z_t)$ is standard Gaussian iid. Introducing the sequence $(a_{2t})$ which is equal to 1 if $|x_t|>0.01$ and $|x_{t-1}|>0.01$, and $a_{2t}=0$ if not, we set $u_t=a_{2t}x_t$. As $(x_t)$ is strictly stationary, it is clear that $P(y_t\neq0)$ is constant. From the specification of $(y_t)$ a dependence is allowed for $(a_t)$ (i.e., we have higher expectation to obtain $y_t=0$ if $y_{t-1}=0$). From a simulation of $(y_t)$ of length $50000$ given in (\ref{ar-1}), we computed $P(y_t\neq0)\approx0.72$. Secondly, to illustrate the asymptotic normality results obtained under the no correlation hypothesis, we set $y_t=\tilde{a}_{2t}x_t$ with $\tilde{a}_{2t}=1$ if $x_t>0.01$, and $\tilde{a}_{2t}=0$ if not. Using again a simulation of length $50000$, we computed $P(y_t\neq0)\approx0.83$. In all the cases, we have $P(a_t=1)=P(a_{1t}=1)P(y_t\neq0)$ as $(y_t)$ and $(a_{1t})$ are taken independent. The possible non-stationary effects for $P(a_t=1)$ will be handled by $P(a_{1t}=1)=g(t/n)$ with

\begin{eqnarray*}
g(s)&=&\delta_31_{\{(0,0.4]\}}(s)+\left[\left(\frac{(\delta_4-\delta_3)}{0.2}\right)s+3\delta_3-2\delta_4\right]\times1_{\{(0.4,0.6]\}}(s)\nonumber\\
&+&\delta_4\times1_{\{(0.6,1]\}}(s)\nonumber
\end{eqnarray*}
The parameters $\delta_1$, $\delta_2$ and $\delta_3$ and $\delta_4$ give the cases of interest for our Monte Carlo experiments:

\begin{itemize}
\item[(i)] Heterogeneous probability, heteroscedastic: Both $P(a_t=1)$ and $\sigma_t$ are time-varying. We set $\delta_1=0.5$, $\delta_2=2$ and $\delta_3=0.3$, $\delta_4=0.9$.
\item[(ii)] Heterogeneous probability, homoscedastic: $P(a_t=1)$ is time-varying and $\sigma_t$ is constant. We set $\delta_1=\delta_2=1$ and $\delta_3=0.3$, $\delta_4=0.9$.
\item[(iii)] Homogeneous probability, heteroscedastic: $P(a_t=1)$ is constant and $\sigma_t$ is time-varying. We set $\delta_1=0.5$, $\delta_2=2$ and $\delta_3=\delta_4=0.6$.
\item[(iv)] Homogeneous probability, homoscedastic: $P(a_t=1)$ and $\sigma_t$ are constant. We set $\delta_1=\delta_2=1$ and $\delta_3=\delta_4=0.6$.
\end{itemize}

In all the experiments, $N=3000$ independent trajectories of (\ref{dgp}) are simulated, with lengths $n=500, 1500$. Such sample sizes roughly correspond to 2 and 6 years of daily returns. All the settings are inspired from the real data study below. The outputs of our simulations are provided in box-plots in Figures \ref{fig-sim-1}-\ref{fig-sim-2bisbis} for the studied autocorrelations, and in Figures \ref{fig-sim-ind-1}-\ref{fig-sim-ind-2} for the index introduced in Section \ref{var-prob}. In Table \ref{conf-int}, the relative frequencies of "pr" and "vpr" autocorrelations outside their 5\% level asymptotic confidence bounds are given. The bandwidths minimize the leave-one-out cross validation criterion (LOOCV) for the estimation of the probability structures $P(a_ta_{t-h}=1)$, $h=1$:

\begin{equation*}
CV_h(b_0)=\sum_{t=1}^{n}\left(\hat{p}_{t}^{(-t)}-a_ta_{t-h}\right)^2,\:
CV_h(b_h)=\sum_{t=1}^{n}\left(\hat{p}_{t,t-h}^{(-t)}-a_ta_{t-h}\right)^2,
\end{equation*}
where $\hat{p}_{t,t-h}^{(-t)}$ is the estimator of $P(a_ta_{t-h}=1)$ using $b_h$, and obtained by omitting $a_ta_{t-h}$.\\

We begin with the situations where the "pr" and "vpr" corrections are both valid (i.e., the (ii), (iii) and (iv) cases, see Figure \ref{fig-sim-ind-2}). 
Its turns out that the "pr" and "vpr" corrections seem unbiased and similar. Indeed, in such cases the $\hat{\rho}_{pr}(1)$ and $\hat{\rho}_{vpr}(1)$ are asymptotically valid for estimating $\rho_{\epsilon}(1)$. 
At first sight there is no loss of efficiency for the more sophisticated "vpr" correction.
Naturally, if there is some evidence that either $P(a_t=1)$ or $\sigma_t$ is constant, then it is better to use the simpler "pr" correction. Since $\hat{\rho}_{pr}(1)$ is relevant in the (ii), (iii) and (iv) cases, it can be seen that the $\hat{\kappa}_1$ converges to zero as the samples are increased (see Figure \ref{fig-sim-ind-1}).

Now, case (i) is examined (i.e., the presence of both heterogeneous variance and non-zero returns probability, see Figure \ref{fig-sim-ind-1}). It can be observed that only the "vpr" correction, seems to handle adequately the double source of non-stationarity. This is in accordance with the fact, that the numerator and the denominator of the serial correlations are computed separately, which induces some undesirable combined non-stationary effects for the $\hat{\rho}_{pr}(1)$ estimator. Such effects cannot be taken into account, when only probabilities are considered for correcting the serial correlations. Moreover, the $\hat{\rho}_{vpr}(1)$ seems to benefit from
the additional piece of information brought from the non-constant $\sigma_t$, on the contrary to the $\hat{\rho}_{pr}(1)$. Indeed, it is found that the $\hat{\rho}_{vpr}(1)$ is more accurate than the  $\hat{\rho}_{pr}(1)$. 
Noting that only $\hat{\rho}_{vpr}(1)$ is valid for estimating $\rho_{\epsilon}(1)$, we can observe that $\hat{\kappa}_1$ does not converges to zero.

In this part, some comments on the confidence intervals for the studied autocorrelations are made. From Table \ref{conf-int}, we can see that the results are the same for the "vpr" and "pr" corrections. This is not surprising as the corrected serial correlations are scale transformation of the correlations that include the zero returns. In addition, we note some finite sample distortions. This is a consequence of the Heteroscedastic Consistent (HC) covariance matrix estimation used to take into account for the non-linearities in the data. Such finite sample shortcomings for the HC estimator is well known in the literature (see e.g. Vilasuso (2001) or Ra\"{\i}ssi (2011)). Simulations for the "pr" correction, not displayed here, show that the relative frequencies of the serial correlations outside their confidence bounds become close to the asymptotic nominal level 5\% for large samples.\\

From our simulation results, we can conclude that if the stock features are incompletely taken into account, then the analysis of the serial correlation between non-zero returns can be misleading. The combined effects of a non-constant zero returns probability and an unconditional heteroscedasticity may produce important bias and a lack of accuracy. However, it is also important to identify each case properly to avoid a correction that is too sophisticated. For these reasons, we recommend consider the $\hat{\kappa}_m$ index of the studied stocks and any event in the company history that can cause a structural change in both the variance and the degree of liquidity.

\subsection{Real data analysis}
\label{real-analysis}

The correlation structures of the non-zero returns of the Chilean stocks presented in the Introduction are studied in detail. The main stocks indexes of the Santiago Stock Exchange (SSE), are the IPSA (the 30 most liquid stocks) and the IGPA (comprising the 30 stocks of the IPSA, plus 48 other stocks according to some criteria). Note that most of the stocks that are part of the IGPA, but not of the IPSA have zero returns. The studied stocks are representative examples of such illiquid assets. The sample sizes are given in Table \ref{stock-corr}.

We first determine the adequate correction for each stock. Recall that the stocks returns, together with the Nadaraya-Watson estimators of $P(a_t=1)$ are displayed in Figures \ref{ns-stocks} and \ref{s-stocks}. In addition, as underlined in the Introduction, events in some companies history strongly suggest that at least $P(a_t=1)$ is not constant. On the other hand, $\hat{\kappa}_1$ is provided in Table \ref{stock-corr}. It emerges that $P(a_t=1)$ seems to be time-varying for most of the stocks, and in particular, the "vpr" correction is advisable for the Cintac and Molymet stocks. We conjecture that in many cases the long-run or abrupt changes in the liquidity degree is not accompanied by a change of the unconditional variance. In particular, it seems that the take-over bid of Metlife on Provida on one hand, and the merger by absorption and capital increase of the Security bank on the other hand, only led to a change in the degree of liquidity. However, it appears that the capital increase of 2010 in the Molymet case resulted in an increase of both the unconditional variance and the probability of a daily price change. 

Now, we turn to the analysis of the first order linear dynamics for the studied stocks. The autocorrelations that seem relevant are in bold type in Table \ref{stock-corr}. The non-standard "vpr" statistics
are close to the "pr" correction when $\hat{\kappa}_1$ is close to zero. Surprisingly, $\hat{\rho}_{pr}(1)$and $\hat{\rho}_{vpr}(1)$ are similar for the Cintac stock, although $\hat{\kappa}_1$ is elevated. For the non-stationary Molymet stock, the "vpr" correction deliver a higher first order correction than the "pr" correction. 
Below, some concluding remarks are provided in the light of the Monte Carlo experiments results and the outputs in Table \ref{stock-corr}.

\section{Conclusion}
\label{concl}

At first glance, for estimating the linear relationship between present and lagged non-zero returns, it makes sense to consider a lack of a daily price movement as missing, without any further processing. In the stationary case, this can be done by applying a scale correction to the usual autocorrelations, using only the probabilities related to observing non-zero returns, as advocated in the classical literature. It turns out that such treatment of the data, is still valid when a structural change can be observed in either the daily non-zero returns probability, or the unconditional variance.
However, when both unconditional heteroscedasticity and non-constant liquidity levels are present, the above correction for the serial autocorrelations is possibly distorted in two ways. Firstly, a bias can be observed. Secondly, a relatively high variability can be noticed. This spurious picture of the serial correlations, can be the consequence of neglected structural changes in the variance and in the liquidity degree. These shortcomings could result in an erroneous analysis of the stock's properties. We propose to solve such issues through a simple scale correction, that takes into account the presence of such a non-stationary behavior in the stock's history. This more realistic correction can lead to accurate assessment of the serial correlations of illiquid assets.

\section*{References}
\begin{description}
\item[] {\sc Bolerslev, T.} (2008) Glossary to ARCH (GARCH). CREATES Research Paper 2008-49.
\item[]{\sc Bondon, P., and Bahamonde, N.} (2012) Least squares estimation of ARCH models with missing observations. \textit{Journal of Time Series Analysis} 33, 880-891.
\item[] {\sc Boudoukh, J., Richardson, M., and Whitelaw, R.} (1994) A tale of three schools: insights on autocorrelations of short-horizon returns. \textit{Review of Financial Studies} 7, 539-573.
\item[]{\sc Brockwell, P.J., and Davis, R.A.} (2006) \textit{Times Series: Theory and Methods}. 2nd edition, Springer, New York.
\item[]{\sc Cavaliere, G., and Taylor, A.M.R.} (2007) Time-transformed unit-root tests for models with non-stationary volatility. \textit{Journal of Time Series Analysis} 29, 300-330.
\item[]{\sc Cavaliere, G., and Taylor, A.M.R.} (2008) Bootstrap Unit Root Tests for Time Series with Nonstationary Volatility. \textit{Econometric Theory} 24, 43-71.
\item[]{\sc Dahlhaus, R.} (1997) Fitting time series models to nonstationary processes. \textit{Annals of Statistics} 25, 1-37.
\item[]{\sc Dalla, V., Giraitis, L., and Phillips, P.C.B.} (2020) Robust tests for white noise and cross-correlation. \textit{Econometric Theory} 1-29.
\item[]{\sc Davidson, J.} (1994) \textit{Stochastic limit theory.} Oxford University Press. New York.
\item[]{\sc Davydov, Y. A.} (1968) Convergence of distributions generated by
stationary stochastic process. \textit{Theory of Probability and its Applications} 13, 691-696.
\item[]{\sc Dunsmuir, W., and Robinson, P. M.} (1981) Estimation of time series models in the presence of missing data. \textit{Journal of the American Statistical Association} 76, 560-568.
\item[]{\sc Engle, R.F.} (1982) Autoregressive conditional heteroscedasticity with estimates of the variance of United Kingdom inflation. \textit{Econometrica} 50, 987-1007.
\item[]{\sc Engle, R.F., and Rangel, J.G.} (2008) The spline GARCH model for unconditional volatility and its global macroeconomic causes. \textit{Review of Financial Studies} 21, 1187-1222.
\item[] {\sc Francq, C., and Zako\"{i}an, J-M.} (2019) \textit{GARCH models : structure, statistical inference, and financial applications.} Wiley.
\item[]{\sc Hafner, C.M., and Linton, O.} (2010) Efficient estimation of a multivariate multiplicative volatility model. \textit{Journal of Econometrics} 159, 55-73.
\item[]{\sc Hansen, B.E.} (1991) Strong laws for dependent heterogeneous processes. \textit{Econometric Theory} 7, 213-221.
\item[]{\sc Harris, D., and Kew, H.} (2014) Portmanteau autocorrelation tests under Q-dependence and heteroscedasticity. \textit{Journal of Time Series Analysis} 35, 203-217.
\item[]{\sc  Hurn, S., Martin, V.L., Phillips, P.C.B. and Yu, J.} (2021) \textit{Financial Econometric Modeling} Oxford University Press.
\item[] {\sc K\"{o}chling, G., M\"{u}ller, J., and Posch, P.N.} (2019) Does the introduction of futures improve the efficiency of Bitcoin? \textit{Finance Research Letters} 30, 367-370.
\item[] {\sc Lesmond, D.A.} (2005) Liquidity of emerging markets. \textit{Journal of Financial Economics} 77, 411-452.
\item[] {\sc Lobato, I.N., Nankervis, J.C., and Savin, N.E.} (2002) Testing for zero autocorrelation in the presence of statistical dependence. \textit{Econometric Theory} 18, 730-743.
\item[]{\sc McNeil, A., Frey, R., and Embrechts, P.} (2015) \textit{Quantitative Risk Management: Concepts, Techniques and Tools}. Princeton University Press. New Jersey.
\item[]{\sc Mikosch, T., and St\u{a}ric\u{a}, C.} (2004) Nonstationarities in financial time series, the long-range dependence, and the IGARCH effects. \textit{Review of Economics and Statistics} 86, 378-390.
\item[]{\sc Parzen, E.} (1963) On spectral analysis with missing observations and amplitude modulation. \textit{Sankhya Series A} 25, 383-392.
\item[]{\sc Patilea, V., and Ra\"{i}ssi, H.} (2012) Adaptive estimation of vector autoregressive models with time-varying variance: application to testing linear causality in mean. \textit{Journal of Statistical Planning and Inference} 142, 2891-2912.
\item[]{\sc Patilea, V., and Ra\"{i}ssi, H.} (2013) Corrected portmanteau tests for VAR models with time-varying variance. \textit{Journal of Multivariate Analysis} 116, 190-207.
\item[] {\sc Patilea, V., and Ra\"{i}ssi, H.} (2014) Testing second order dynamics for autoregressive processes in presence of time-varying variance. \textit{Journal of the American Statistical Association} 109, 1099-1111.
\item[] {\sc Patilea, V., and Ra\"{i}ssi, H.} (2021) Powers correlation analysis of returns with a nonstationary zero-process. Arxiv Working document https:$//$arxiv.org$/$pdf$/$2104.04472.pdf.
\item[]{\sc Phillips, P.C.B.} (1987) Time series regression with a unit root. \textit{Econometrica} 55, 277-301.
\item[]{\sc Phillips, P.C.B., and Xu, K.L.} (2006) Inference in autoregression under heteroskedasticity. \textit{Journal of Time Series Analysis} 27, 289-308.
\item[]{\sc Ra\"{\i}ssi, H.} (2011) Testing linear causality in mean when the number of estimated parameters is high. \textit{ Electronic Journal of Statistics} 5, 507-533.
\item[]{\sc Romano, J. P., and Thombs, L. A.} (1996) Inference for autocorrelations under weak assumptions. \textit{Journal of the American Statistical Association} 91, 590-600.
\item[]{\sc St\u{a}ric\u{a}, C., and Granger, C.} (2005) Nonstationarities in stock returns. \textit{Review of Economics and Statistics} 87, 503-522.
\item[]{\sc Stoffer, D.S., and Toloi, C.M.C.} (1992) A note on the Ljung-Box-Pierce portmanteau statistic with missing data. \textit{Statistics and Probability Letters} 13, 391-396.
\item[]{\sc Truquet, L.} (2019) Local stationarity and time-inhomogeneous Markov chains. \textit{The Annals of Statistics} 47, 2023-2050.
\item[]{\sc Vilasuso, J.} (2001) Causality tests and conditional heteroscedasticity: Monte Carlo evidence. \textit{Journal of Econometrics} 101, 25-35.
\item[]{\sc Wang, S., Zhao, Q., and Li, Y.} (2019) Testing for no-cointegration under time-varying variance. \textit{Economics Letters} 182, 45-49.
\item[]{\sc Xu, K.L., and Phillips, P.C.B.} (2008) Adaptive estimation of autoregressive models with time-varying variances. \textit{Journal of Econometrics} 142, 265-280.
\end{description}

\newpage

\section*{Proofs}

In the sequel, $C>0$ is a constant, possibly different from line to line.

\begin{proof}[Proof of Proposition \ref{propostu2}]
Note that $(r_t)$ and $(a_t)$ are strictly stationary ergodic under the conditions of Proposition \ref{propostu2}. Then, from the ergodic Theorem (see Francq and Zako\"{\i}an (2019) Theorem A.2), we have
\begin{equation}\label{eq0p2}
\hat{\gamma}_0(h)\stackrel{a.s.}{\longrightarrow}\gamma_{0}(h):=E(r_tr_{t-h}),
\end{equation}
for any $h=0,1,\dots,m$. Hence, we clearly obtain
\begin{equation}\label{eq1p2}
\widehat{\Gamma}_0(m)\stackrel{a.s.}{\longrightarrow}\Gamma_{0}(m),
\end{equation}
where we recall that $\Gamma_{0}(m):=(\rho_{0}(1),\dots,\rho_{0}(m))'$. Using again the Ergodic Theorem, we write
\begin{equation}\label{eq2p2}
\hat{\gamma}_a(h)\stackrel{a.s.}{\longrightarrow}P(a_ta_{t-h}=1),
\end{equation}
for any $h=0,1,\dots,m$. In view of (\ref{eq1p2}), (\ref{eq2p2}) and (\ref{auto-stat-case}) the result follows.
\end{proof}

\begin{proof}[Proof of Proposition \ref{propostu2bis}]
First note that $E(\Upsilon_t)=0$, since $(r_t)$ is uncorrelated. In addition, we have $E(\Upsilon_{it}\Upsilon_{jt})<\infty$ for any $1\leq i,j\leq m$. Then, noting that $(\Upsilon_t)$ is a strictly stationary ergodic martingale increments sequence, we have

$$n^{-\frac{1}{2}}\sum_{t=m+1}^{n}\Upsilon_t\stackrel{d}{\longrightarrow}\mathcal{N}(0,\Sigma_{\Upsilon}),$$
from the Lindeberg CLT (see Francq and Zako\"{\i}an (2019) Corollary A.1). Hence, considering (\ref{eq0p2}) with $h=0$, (\ref{eq2p2}) and the Slutsky Lemma, the desired result follows.
\end{proof}

\begin{proof}[Proof of Proposition \ref{propostu3}]
We begin with the numerator of $\hat{\rho}_0(h)$. Denoting by $\alpha_{\Upsilon}$ the $\alpha$-mixing coefficients of $(\Upsilon_t)$, we have $\alpha_{\Upsilon}(|k|+1)\leq\alpha_{\epsilon}(|k|-m)$, setting $\alpha_{\epsilon}(k)=1/4$, for $k<0$.
In view of the above argument, Assumption \ref{mixing-an} implies that $\sum_{k=0}^{\infty}\left(\alpha_{\Upsilon}(k)\right)^{1-1/2\mu_2}<\infty$. In addition, by the H\"{o}lder inequality, we write $\|\Upsilon_t\|_{1+\nu}<\infty$. Then, using the Strong Law of Large Numbers (SLLN) of Hansen (1991), corollary 4, we obtain
$$\hat{\gamma}_0(h)\stackrel{a.s.}{\longrightarrow}\lim_{n\to\infty}n^{-1}\sum_{t=1+h}^{n}E(r_tr_{t-h}),$$
for $h=0,1,\dots,m$. From Assumption \ref{stationarity}, we have

$$n^{-1}\sum_{t=1+h}^{n}E(r_tr_{t-h})=E(\epsilon_t\epsilon_{t-h}|a_ta_{t-h}=1)
\left\{n^{-1}\sum_{t=1+h}^{n}\sigma_t\sigma_{t-h}P(a_ta_{t-h}=1)\right\}.$$
In view of Assumption \ref{tv-var} and \ref{tv-prob}, 
deduce that
\begin{equation}\label{maipu}
n^{-1}\sum_{t=1+h}^{n}\sigma_t\sigma_{t-h}P(a_ta_{t-h}=1)=\int_0^1v^2(s)g_h(s)ds+R_n,
\end{equation}
where the remainder $R_n$ is non stochastic, such that $|R_n|<Cn^{-1}$. Then
$$\hat{\gamma}_0(h)\stackrel{a.s.}{\longrightarrow}\rho_{\epsilon}(h)\int_0^1v^2(s)g_h(s)ds.$$
Similarly, for the denominator we have
\begin{equation}\label{eq3p3}
\hat{\gamma}_0(0)\stackrel{a.s.}{\longrightarrow}\int_0^1v^2(s)g_0(s)ds,
\end{equation}
in view of Assumption \ref{identifiability}. Deduce that
\begin{equation}\label{eq1p3}
\hat{\rho}_0(h)\stackrel{a.s.}{\longrightarrow}\rho_{\epsilon}(h)\frac{\int_0^1v^2(s)g_h(s)ds}{\int_0^1v^2(s)g_0(s)ds}.
\end{equation}

Now, the asymptotic behavior of the kernel estimators is established. Here, as a finite number of breaks is assumed, we consider a smooth function $g_0(\cdot)$ without a loss of generality. Also, we only consider $\hat{p}_{t}$, as the arguments are similar for $\hat{p}_{t,t-h}$. Let $\bar{p}_t=\sum_{j=1}^{n} w_{tj}(b)P(a_j=1)$. From Lemma A(c) in Xu and Phillips (2008), we write
\begin{eqnarray*}
\hat{p}_{t}-\bar{p}_t&=&\sum_{j=1}^{n} w_{tj}(b)(a_j-P(a_j=1))\\
&=&\left[\frac{1}{nb}\sum_{j=1}^{n}K_{tj}(a_j-E(a_j))\right]\times(1+o(1)).
\end{eqnarray*}
On the other hand, we have
\begin{eqnarray*}
E\left(\frac{1}{nb}\sum_{j=1}^{n}K_{tj}(a_j-E(a_j))\right)^2
&=&\frac{1}{(nb)^2}\sum_{i,j=1}^{n}K_{tj}K_{ti}E\{(a_j-E(a_j))(a_i-E(a_i))\}\\
&\leq&{\small\left(\frac{R^2}{nb^2}\right)n^{-1}\sum_{|h|<n}\sum_{i=1+h}^{n}\left|E\{(a_i-E(a_i))(a_{i-h}-E(a_{i-h}))\}\right|}\\
&\leq& \frac{\widetilde{C}}{nb^2} \sum_{|h|<n}\frac{n-h}{n}\alpha_\epsilon(|h|)^{1-1/2\mu_2}\\
&=&O\left(\frac{1}{nb^2}\right),
\end{eqnarray*}
in view of our mixing assumption. The last inequality is obtained from the inequality of Davydov (1968): there exists a constant $C$ such that

$$|E\{(a_j-E(a_j))(a_i-E(a_i))\}|<C\|(a_j-E(a_j))\|_{4\mu_2}\|(a_i-E(a_i))\|_{4\mu_2}\alpha_\epsilon(|h|)^{1-1/2\mu_2}.$$
In view of the above, we obtain
$$\hat{p}_{t}-\bar{p}_t=O_p\left(\frac{1}{\sqrt{n}b}\right).$$
Next, let us write
\begin{eqnarray*}
&&\frac{1}{nb}\sum_{i=1}^nK_{[nr]i}P(a_i=1)\\&=&\frac{1}{nb}\sum_{i=1}^{n}K\left(\frac{[nr]-i}{nb}\right)g_0\left(\frac{i}{n}\right)
\\&=&\frac{1}{b}\left[\int_{1/n}^{2/n}K\left(\frac{[nr]-[ns]}{nb}\right)g_0\left(\frac{[ns]}{n}\right)ds+\dots\right.\\
&&+\left.\int_{1}^{(n+1)/n}K\left(\frac{[nr]-[ns]}{nb}\right)g_0\left(\frac{[ns]}{n}\right)ds+O\left(\frac{1}{n}\right)\right]\\&&
=\frac{1}{b}\left(\int_{1/n}^{(n+1)/n}K\left(\frac{[nr]-ns}{nb}\right)g_0\left(s\right)ds\right)+O\left(\frac{1}{nb}\right)\\&&
\stackrel{z=(s-r)/b}{=}\int_{(\frac{1}{n}-r)/b}^{(1+\frac{1}{n}-r)/b}
K\left(\frac{[nr]-n(r+bz)}{nb}\right)g_0\left(r+bz\right)dz+O\left(\frac{1}{nb}\right)\\&&
=\int_{-\infty}^{\infty}
K\left(\frac{[nr]-nr}{nb}-z\right)g_0\left(r+bz\right)dz+O\left(\frac{1}{nb}\right),
\end{eqnarray*}
where the last equality is obtained for small enough $b$, and since a compact support is assumed for $K(\cdot)$.
Hence, using the Lipschitz condition in Assumption \ref{tv-prob}, deduce that

\begin{equation*}
\frac{1}{nb}\sum_{i=1}^nK_{[nr]i}p_i=g_0(r)+O(b)+O\left(\frac{1}{nb}\right).
\end{equation*}
In view of the above, we may write
\begin{equation}\label{eq11}
\hat{p}_{t}=P(a_t=1)+O_p\left(\frac{1}{\sqrt{n}b}\right)+O(b).
\end{equation}
Similarly, we also have
\begin{equation}\label{eq22}
\hat{p}_{t,t-h}=P(a_ta_{t-h}=1)+O_p\left(\frac{1}{\sqrt{n}b}\right)+O(b).
\end{equation}
Considering the denominator of the serial correlations correcting term, we obtain
\begin{eqnarray}
\hat{\gamma}_{ar^2}(h)&=&(n-h)^{-1}\sum_{t=1+h}^{n}r_t^2\frac{\hat{p}_{t,t-h}}{\hat{p}_{t}}\nonumber\\
&=&(n-h)^{-1}\sum_{t=1+h}^{n}r_t^2\left(\frac{\hat{p}_{t,t-h}}{\hat{p}_{t}}-\frac{P(a_ta_{t-h}=1)}{P(a_t=1)}\right)\nonumber\\
&+&(n-h)^{-1}\sum_{t=1+h}^{n}r_t^2\left(\frac{P(a_ta_{t-h}=1)}{P(a_t=1)}\right).\label{xxxxxsf}
\end{eqnarray}
From (\ref{eq11}) and (\ref{eq22}), and using the SLLN of Hansen (1991), the first term of the above equality converges to zero in probability. On the other hand, using again the SLLN of Hansen (1991), we have

$$(n-h)^{-1}\sum_{t=1+h}^{n}r_t^2\left(\frac{P(a_ta_{t-h}=1)}{P(a_t=1)}\right)
\stackrel{a.s.}{\longrightarrow}\int_0^1v^2(s)g_h(s)ds.$$
Hence, we obtain
\begin{equation}\label{lluvia}
\hat{\gamma}_{ar^2}(h)=\int_0^1v^2(s)g_h(s)ds+o_p(1).
\end{equation}
From (\ref{eq3p3}), (\ref{eq1p3}) and (\ref{lluvia}) the desired result follows.
\end{proof}

\begin{proof}[Proof of Proposition \ref{propostu3bis}]
Let $\lambda\in\mathbb{R}^m$. Using similar arguments as in the proof of Proposition \ref{propostu2bis}, then the SLLN of Hansen (1991) can be used and we obtain

$$n^{-1}\sum_{t=1}^{n}E\left(\lambda'\Upsilon_t\Upsilon_t'\lambda|\mathcal{F}_{t-1}\right)
\stackrel{a.s.}{\longrightarrow}\lim_{n\to\infty}n^{-1}\sum_{t=1}^{n}E(\lambda'\Upsilon_t\Upsilon_t'\lambda).$$
Now, let us write $\lambda'\Upsilon_t=v_\lambda(t/n)\xi_{\lambda t}$, where $E(\xi_{\lambda t}^2|\mathcal{A}_t^c)=1$ and $\mathcal{A}_t^c=\{a_t=0\}\cup\{\cap_{i=1}^m\{a_{t-i}=0\}\}$. Then, introducing $\varsigma_\lambda=\sup_{(0,1]}(v^2(r))$ and $\iota_\lambda=\inf_{(0,1]}(v(r))$, we write
{\footnotesize
\begin{eqnarray*}
n^{-1}\sum_{t=1}^{n}E\left((\lambda'\Upsilon_t)^21_{\{|\lambda'\Upsilon_t|\geq \sqrt{n}\varepsilon\}}\right)&\leq&\varsigma_\lambda\left[n^{-1}\sum_{t=1}^{n}E\left(\xi_{\lambda t}^21_{\{\iota_\lambda|\xi_{\lambda t}|\geq \sqrt{n}\varepsilon\}}\right)\right]\\&\leq&\varsigma_\lambda\left[n^{-1}\sum_{t=1}^{n}E\left(\xi_{\lambda t}^21_{\{|\iota_\lambda\xi_{\lambda t}|\geq \sqrt{n}\varepsilon\}}|\mathcal{A}_t^c\right)\right]\\&=&\varsigma_\lambda E\left(\xi_{\lambda t}^21_{\{\iota_\lambda|\xi_{\lambda t}|\geq \sqrt{n}\varepsilon\}}|\mathcal{A}_t^c\right)\to0,
\end{eqnarray*}}
as $n\to\infty$, and in view of Assumption \ref{stationarity}.
Then, from the CLT for martingale differences sequences in Francq and Zakoïan (2019), Theorem A.3, and using (\ref{eq3p3}) and (\ref{lluvia}) the first statement follows. For the second statement, it is easy to see that the SLLN of Hansen (1991) can be applied using the same arguments as in the proof of Proposition \ref{propostu3}.
\end{proof}

\begin{proof}[Proof of Proposition \ref{index-prop}]
From (\ref{xxxxxsf}) we have

\begin{equation}\label{hatilde}
\widehat{\Delta}_h=\widetilde{\Delta}_h+o_p(1).
\end{equation}
where
$$\widetilde{\Delta}_h=\frac{n^{-1}\sum_{t=1}^{n}r_t^2}{(n-h)^{-1}\sum_{t=1+h}^{n}r_t^2\frac{P(a_ta_{t-h}=1)}{P(a_t=1)}}
-\frac{n^{-1}\sum_{t=1}^{n}a_t}{(n-h)^{-1}\sum_{t=1+h}^{n}a_ta_{t-h}}.$$
Some computations give

{\footnotesize\begin{eqnarray*}
\widetilde{\Delta}_h&=&\frac{n^{-1}\sum_{t=1}^{n}r_t^2-E(r_t^2)}{(n-h)^{-1}\sum_{t=1+h}^{n}r_t^2\frac{P(a_ta_{t-h}=1)}{P(a_t=1)}}\\
&+&\frac{\left(n^{-1}\sum_{t=1}^{n}E(r_t^2)\right)\left((n-h)^{-1}\sum_{t=1+h}^{n}a_ta_{t-h}\right)}
{\left((n-h)^{-1}\sum_{t=1+h}^{n}r_t^2\frac{P(a_ta_{t-h}=1)}{P(a_t=1)}\right)\left((n-h)^{-1}\sum_{t=1+h}^{n}a_ta_{t-h}\right)}\\
&-&\frac{\left(n^{-1}\sum_{t=1}^{n}P(a_t=1)\right)(n-h)^{-1}\sum_{t=1+h}^{n}r_t^2\frac{P(a_ta_{t-h}=1)}{P(a_t=1)}}
{\left((n-h)^{-1}\sum_{t=1+h}^{n}r_t^2\frac{P(a_ta_{t-h}=1)}{P(a_t=1)}\right)\left((n-h)^{-1}\sum_{t=1+h}^{n}a_ta_{t-h}\right)}\\
&-&\frac{n^{-1}\sum_{t=1}^{n}a_t-P(a_t=1)}{(n-h)^{-1}\sum_{t=1+h}^{n}a_ta_{t-h}}\\
&:=&A+B+Q+H,\: say.
\end{eqnarray*}}
The numerator of $B+Q$ is equal to

\begin{eqnarray*}
 &&\left(n^{-1}\sum_{t=1}^{n}E(r_t^2)\right)\left((n-h)^{-1}\sum_{t=1+h}^{n}a_ta_{t-h}-P(a_ta_{t-h}=1)\right)\\
&-&\left(n^{-1}\sum_{t=1}^{n}P(a_t=1)\right)\left((n-h)^{-1}\sum_{t=1+h}^{n}(r_t^2-E(r_t^2))\frac{P(a_ta_{t-h}=1)}{P(a_t=1)}\right)\\
&+&\left(n^{-1}\sum_{t=1}^{n}E(r_t^2)\right)\left((n-h)^{-1}\sum_{t=1+h}^{n}P(a_ta_{t-h}=1)\right)\\
&-&\left(n^{-1}\sum_{t=1}^{n}P(a_t=1)\right)\left((n-h)^{-1}\sum_{t=1+h}^{n}E(r_t^2)\frac{P(a_ta_{t-h}=1)}{P(a_t=1)}\right)\\
&:=&D+E+F+G,\:say.
\end{eqnarray*}
Since it is supposed that $\Delta_h=0$, then we can write

$$\int_{0}^{1}v^2(s)g_0(s)ds\int_{0}^{1}g_h(s)ds
=\int_{0}^{1}g_0(s)ds\int_{0}^{1}v^2(s)g_h(s)ds.$$
Hence, using the usual convergence property of Riemann sums, and in view of the Lipschitz and the finite number of breaks conditions in Assumption \ref{tv-var} and \ref{tv-prob}, we obtain $F+G=O(n^{-1}).$
Gathering the above facts for $h\in\{1,\dots,m\}$, deduce that,
\begin{equation}\label{avt-der}
\left(\widehat{\Delta}_1,\dots,\widehat{\Delta}_m\right)'=\Psi \left\{n^{-1}\sum_{t=1+m}^{n}Z_t\right\}+o_p(1)=o_p(1),
\end{equation}
using the SLLN of Hansen (1991), and where

$$\Psi=\left(
           \begin{array}{ccccc}
             \Psi_1 & 0 & 0 & \left(\int_{0}^{1}v^2(s)g_1(s)ds\right)^{-1} & -\left(\int_{0}^{1}g_1(s)ds\right)^{-1} \\
             0 & \ddots & 0 & \vdots & \vdots \\
             0 & 0 & \Psi_m & \left(\int_{0}^{1}v^2(s)g_m(s)ds\right)^{-1} & -\left(\int_{0}^{1}g_m(s)ds\right)^{-1}\\
           \end{array}
         \right),$$
is a $m\times(2m+2)$ dimensional matrix, with

$$\Psi_h=\left(\frac{\int_{0}^{1}v^2(s)g_0(s)ds}{\int_{0}^{1}v^2(s)g_h(s)ds\int_{0}^{1}g_h(s)ds},
-\frac{\int_{0}^{1}g_0(s)ds}{\int_{0}^{1}v^2(s)g_h(s)ds\int_{0}^{1}g_h(s)ds}\right).$$
The equation (\ref{avt-der}) is justified, since the $(2m+2)$-dimensional (unobserved) vector
\begin{eqnarray*}
Z_t&=&\left(a_ta_{t-1}-P(a_ta_{t-1}=1),(r_t^2-E(r_t^2))\frac{P(a_ta_{t-1}=1)}{P(a_t=1)},\dots\right.\\&&\left. ,a_ta_{t-m}-P(a_ta_{t-m}=1),(r_t^2-E(r_t^2))\frac{P(a_ta_{t-m}=1)}{P(a_t=1)},\right.
\\&&\left. r_t^2-E(r_t^2),a_t-P(a_t=1)\right)',
\end{eqnarray*}
is such that $E(Z_t)=0$.

%
%
%
\end{proof}

\begin{proof}[Proof of Proposition \ref{index-prop-power}] Using the SLLN of Hansen (1991) for the numerators and denominators in $\widetilde{\Delta}_h$, we have
$$\widetilde{\Delta}_h\stackrel{a.s.}{\longrightarrow}\Delta_h.$$
From (\ref{hatilde}), we get the desired results if at least some $\Delta_h\neq0$, $h\in\{1,\dots,m\}$.
\end{proof}


\newpage

\section*{Tables and Figures}

\begin{table}[hh]\!\!\!\!\!\!\!\!\!\!
\begin{center}
\caption{\small{The sample sizes and zero return probabilities (in \%) of illiquid Santiago financial market stocks. The first order correlations with corrections assuming stationary zero returns occurrence ($\hat{\rho}_{pr}(1)$), and the non stationary variance and zero returns probability correction ($\hat{\rho}_{vpr}(1)$). The index $\hat{\kappa}_1$ is also displayed. The above notations are introduced in Sections \ref{first-order} and \ref{var-prob}. The relevant serial correlations are in bold type, together with their standard deviations.}}
\begin{tabular}{c|c|c||c|c||c|}\cline{2-6}
     & n & 1-$\hat{\gamma}_a(0)$& $\hat{\rho}_{pr}(1)$& $\hat{\rho}_{vpr}(1)$& $\hat{\kappa}_1$ \\
  \hline
  \multicolumn{1}{|c|}{{\small Aguas}}& 1958 &2.40& $\textbf{-0.13}_{0.09}$ & -0.13  & 0.00 \\
  \hline
  \multicolumn{1}{|c|}{{\small Clinica las Condes}} & 1958  & 50.82& $\textbf{0.03}_{0.07}$ & 0.02 & 0.01  \\
  \hline\hline
  \multicolumn{1}{|c|}{{\small Cintac}} & 5159&67.76 & 0.09 & $\textbf{0.11}_{0.06}$ & 0.27 \\
  \hline
  \multicolumn{1}{|c|}{{\small Cruzados}} & 2594 & 70.39   & $\textbf{-0.11}_{0.11}$  & -0.11 & 0.00 \\
  \hline
    \multicolumn{1}{|c|}{{\small HF }} & 1941 & 58.84   & $\textbf{0.41}_{0.09}$ & 0.37 & 0.02 \\
  \hline
    \multicolumn{1}{|c|}{{\small Molymet}}& 4896 & 72.06  & 0.09  &  $\textbf{0.13}_{0.05}$ & 0.28 \\
  \hline
    \multicolumn{1}{|c|}{{\small Paz}} & 3392& 27.49  & $\textbf{0.11}_{0.05}$ & 0.10 & 0.00  \\
  \hline
    \multicolumn{1}{|c|}{{\small Vapores}} & 5162 & 19.96  &  $\textbf{0.13}_{0.02}$ &  0.13 & 0.00 \\
  \hline
  \multicolumn{1}{|c|}{{\small Security}} & 5188 & 37.72  & $\textbf{-0.02}_{0.03}$  &  -0.02 & 0.00 \\
  \hline
  \multicolumn{1}{|c|}{{\small Provida}} & 5179 &  40.08 &  $\textbf{-0.01}_{0.03}$ & -0.01 & 0.00 \\
  \hline
\end{tabular}
\label{stock-corr}
\end{center}
\end{table}

\begin{table}[hh]\!\!\!\!\!\!\!\!\!\!
\begin{center}
\caption{\small{The relative frequencies (in \%) of empirical "pr" and "vpr" serial correlation outside their 5\% level asymptotic confidence bounds (see Proposition \ref{propostu2bis} for the "pr" correction, and Proposition \ref{propostu3bis} for the "vpr" correction).}}
{\footnotesize\begin{tabular}{|c|c|c||c|c||c|c|}\cline{3-7}
   \multicolumn{2}{c|}{  } & $n$ &\multicolumn{2}{c||}{500}&\multicolumn{2}{c|}{1500} \\
   \hline
  prob. & var. & Case & pr & vpr & pr & vpr \\
  \hline\hline
  homo. & homo. & (iv) & 2.9 & 2.9 & 2.8 & 2.8 \\
  \hline
  homo. & het. & (iii) & 2.6 & 2.6 & 2.9 & 2.9 \\
  \hline
  het. & homo. & (ii)  & 3.0 & 3.0 & 3.9 & 3.9 \\
  \hline
  het. & het. & (i)    & 3.3 & 3.3 & 3.7 & 3.7 \\
  \hline
\end{tabular}}
\label{conf-int}
\end{center}
\end{table}

\begin{figure}[h]\!\!\!\!\!\!\!\!\!\!
\vspace*{18cm}

\protect \includegraphics{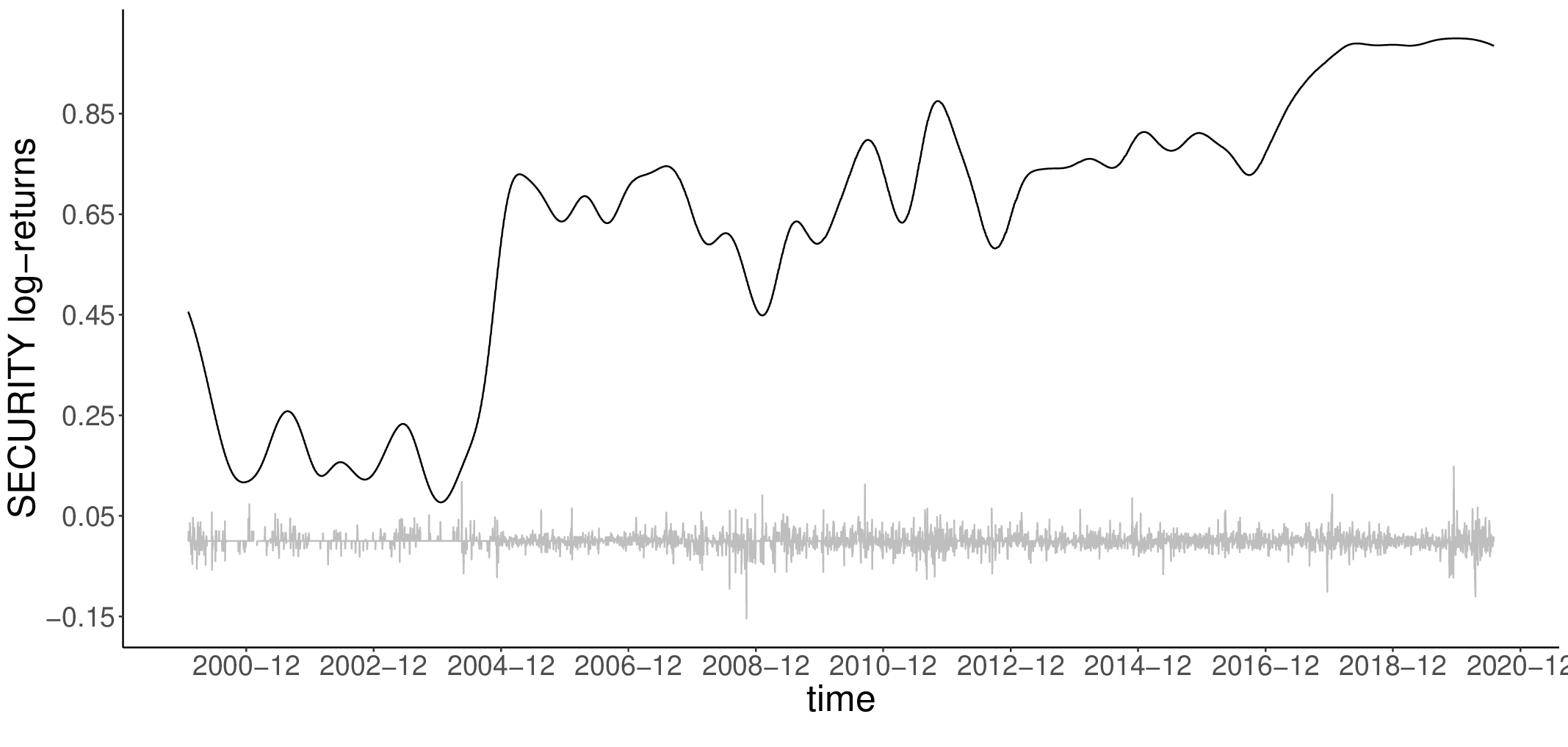}
\protect \includegraphics{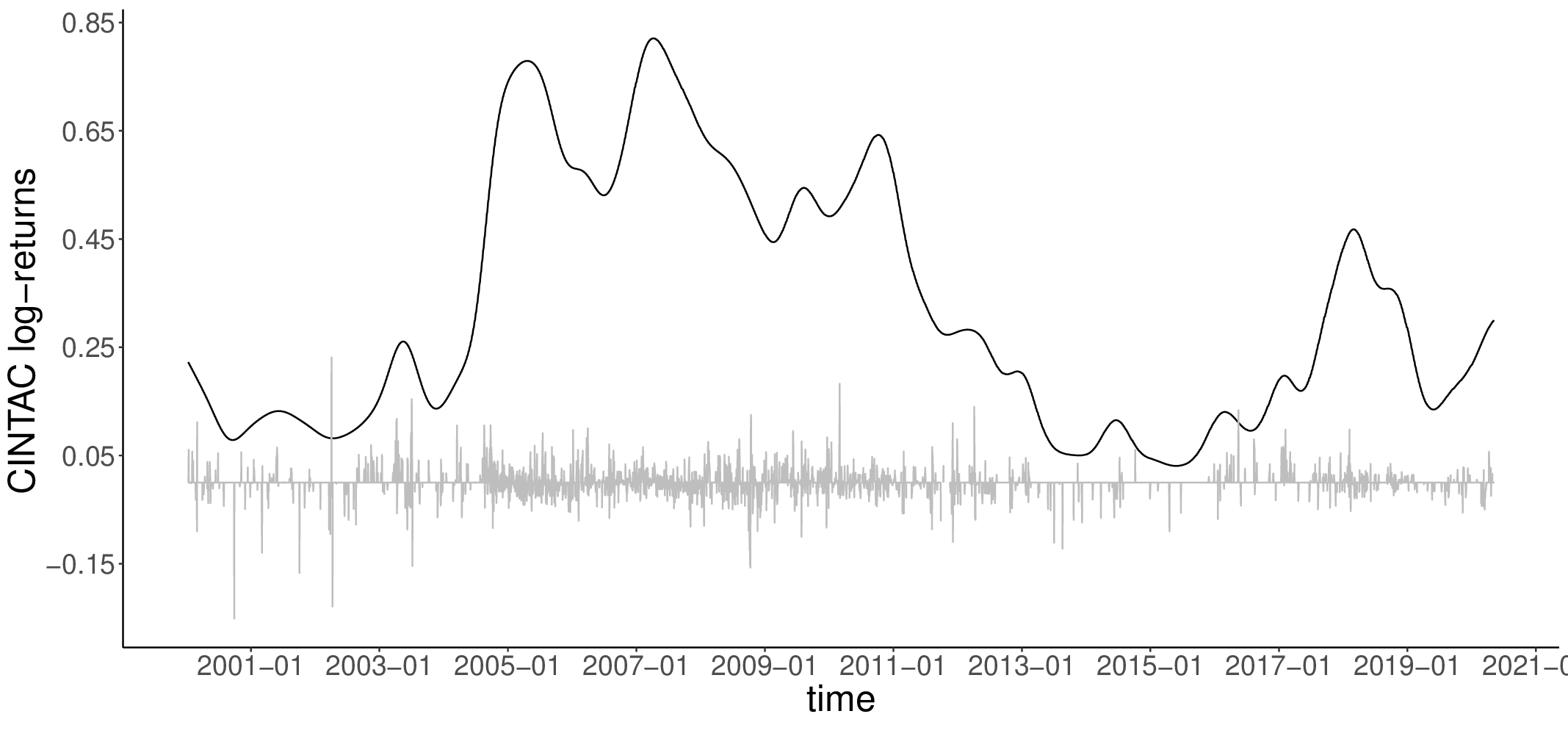}
\protect \includegraphics{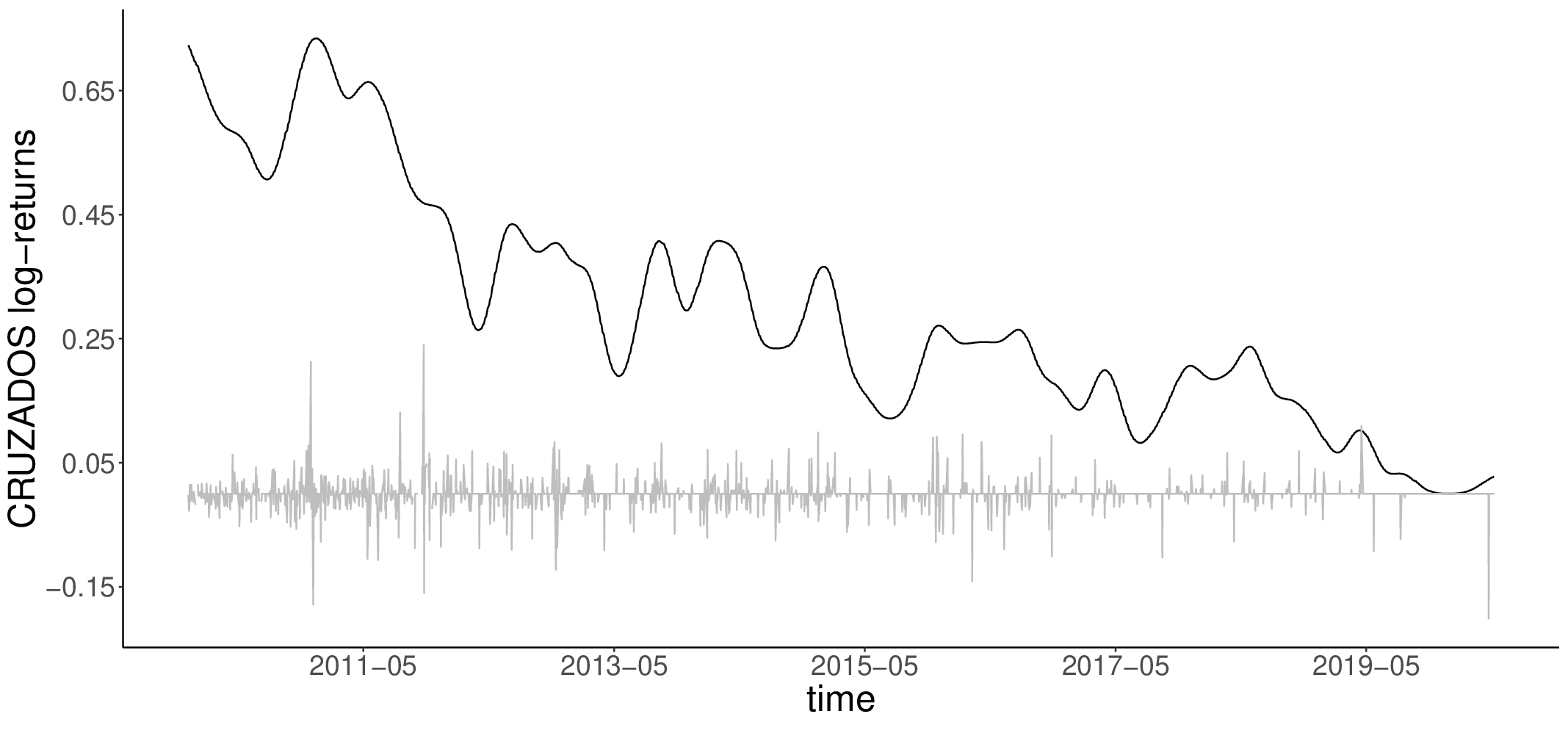}
\protect \includegraphics{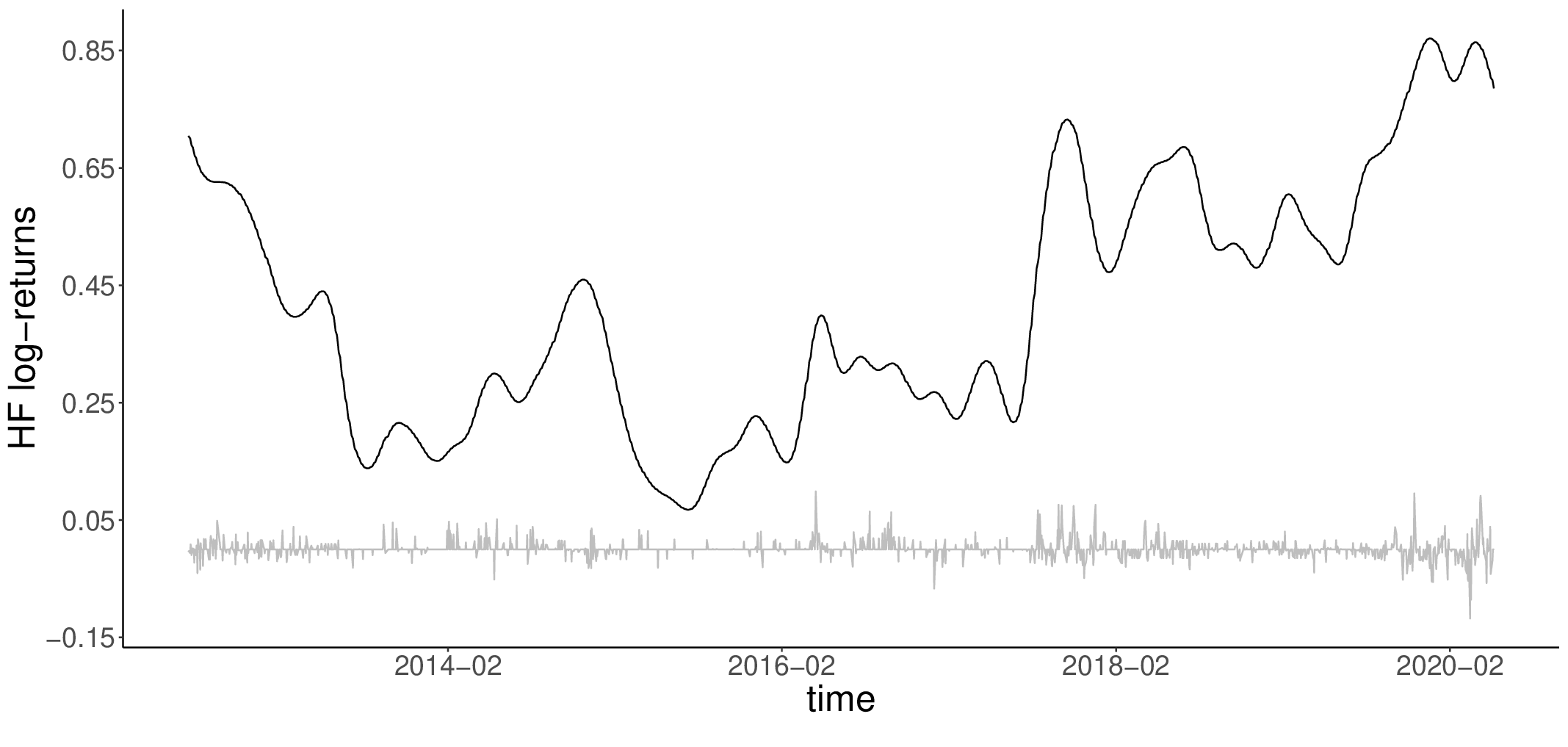}
\protect \includegraphics{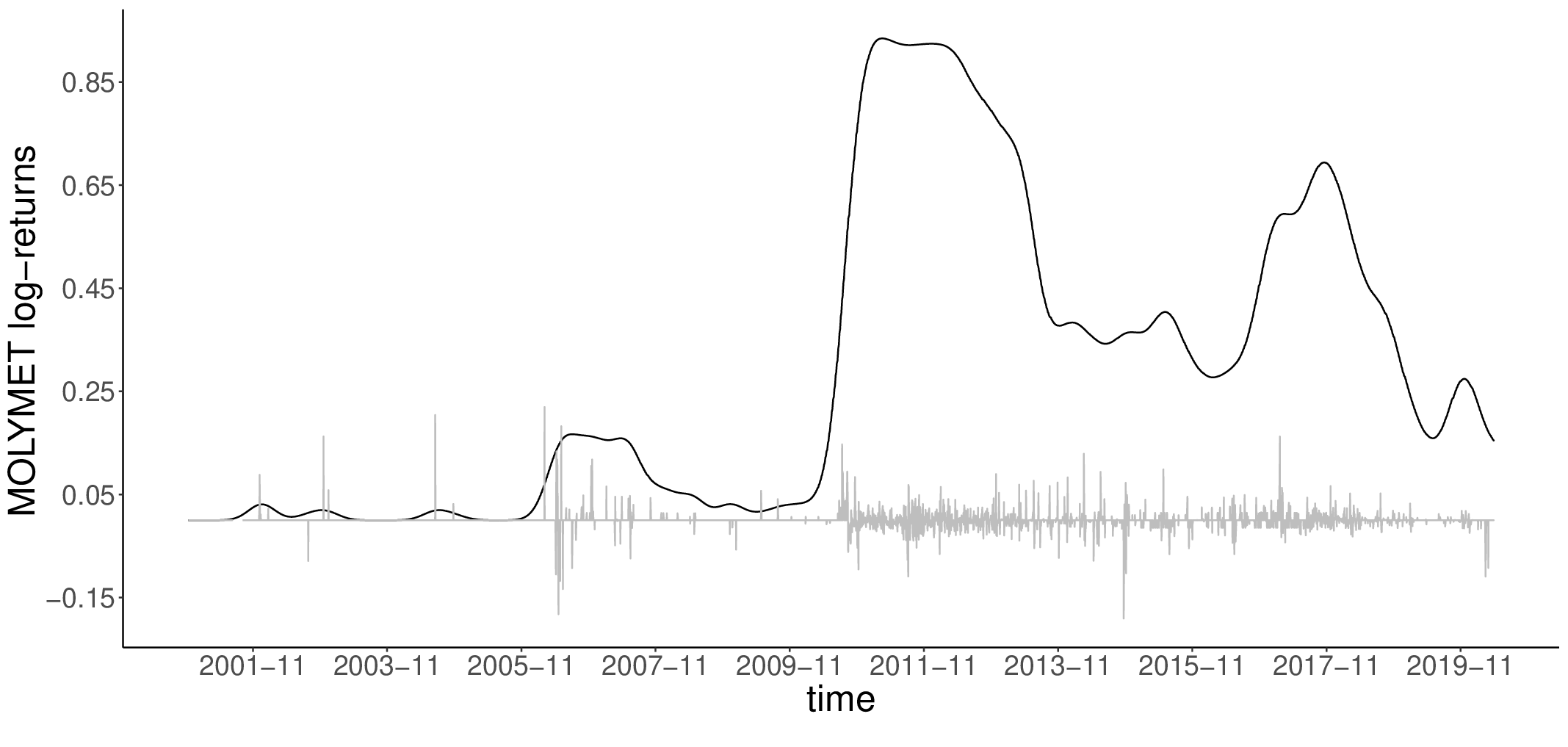}
\protect \includegraphics{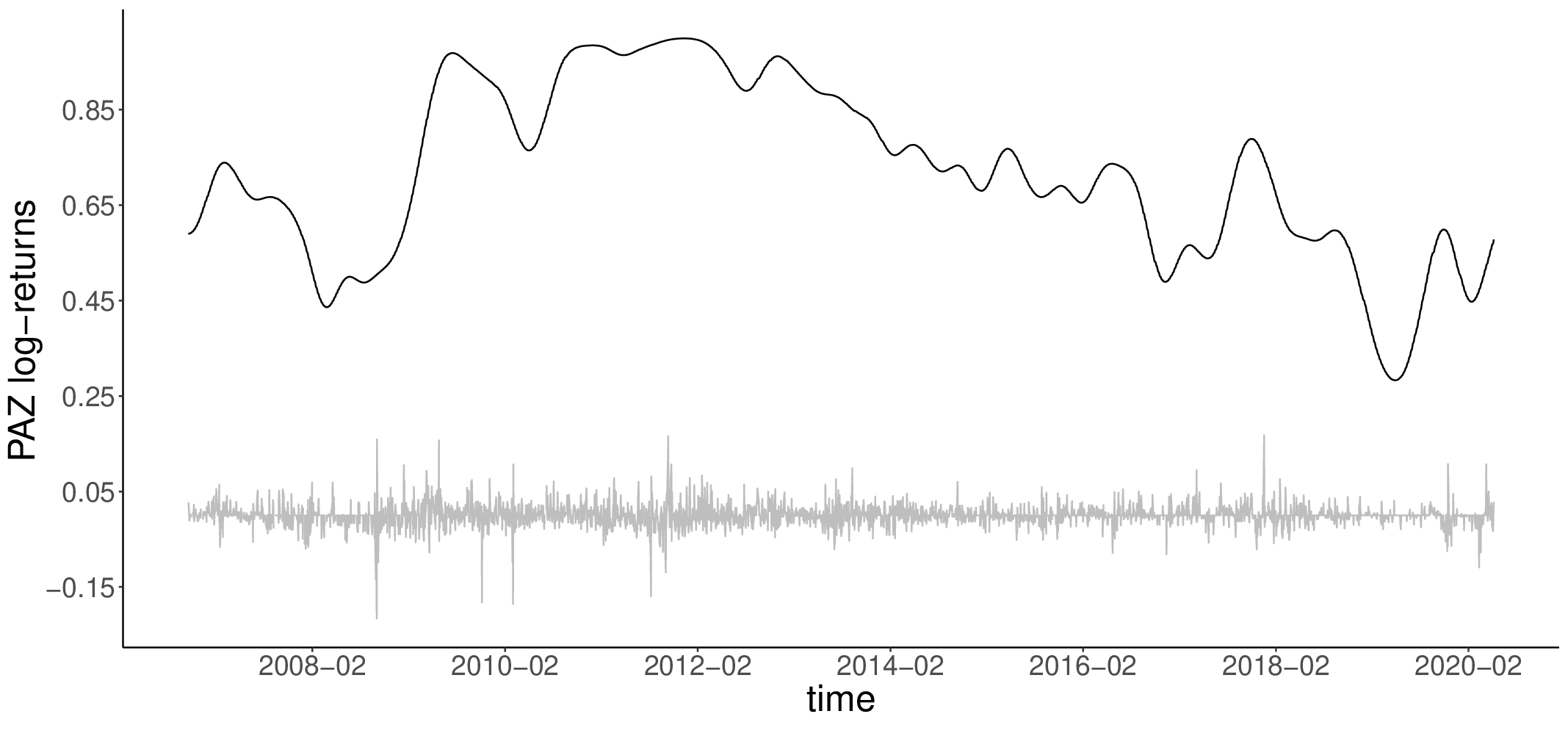}
\protect \includegraphics{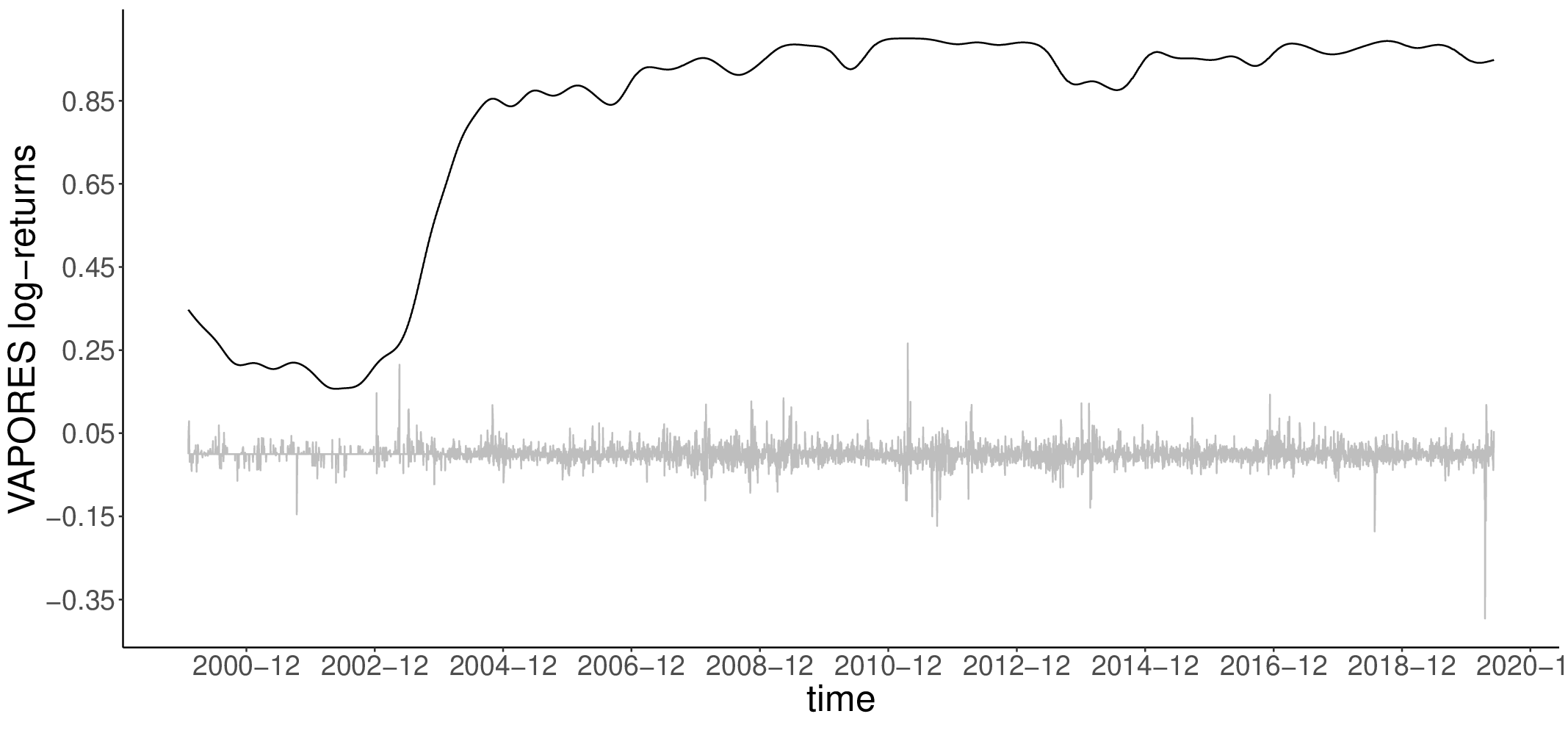}
\protect \includegraphics{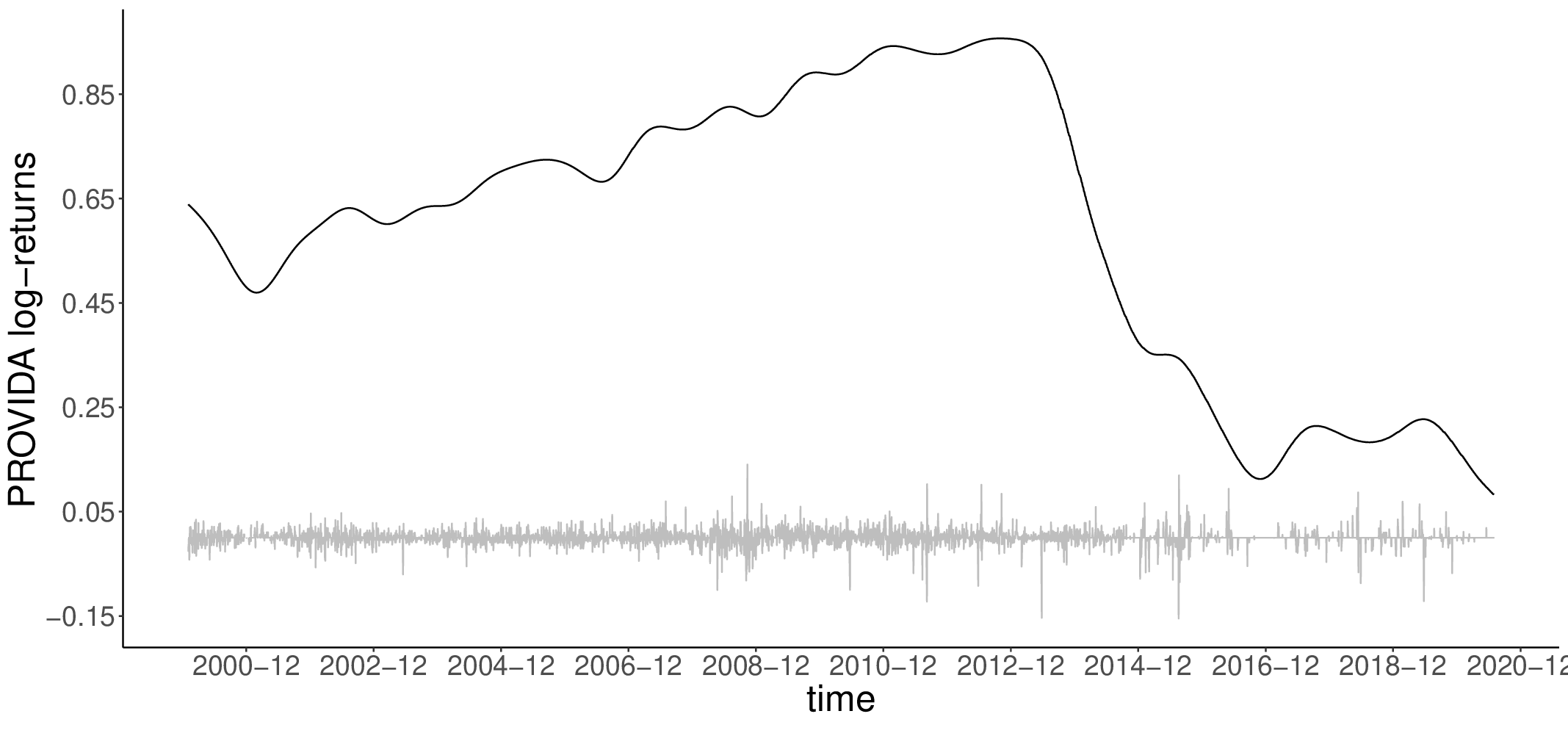}
\caption{\label{ns-stocks}
{\footnotesize The log-returns of different illiquid stocks of the Santiago financial market, which seem to have non-constant daily price change probabilities. The kernel smoothing of the $a_t$ values are displayed in full line. Data source: Yahoo Finance.}}
\end{figure}

\begin{figure}[h]\!\!\!\!\!\!\!\!\!\!
\vspace*{4cm}

\protect \includegraphics{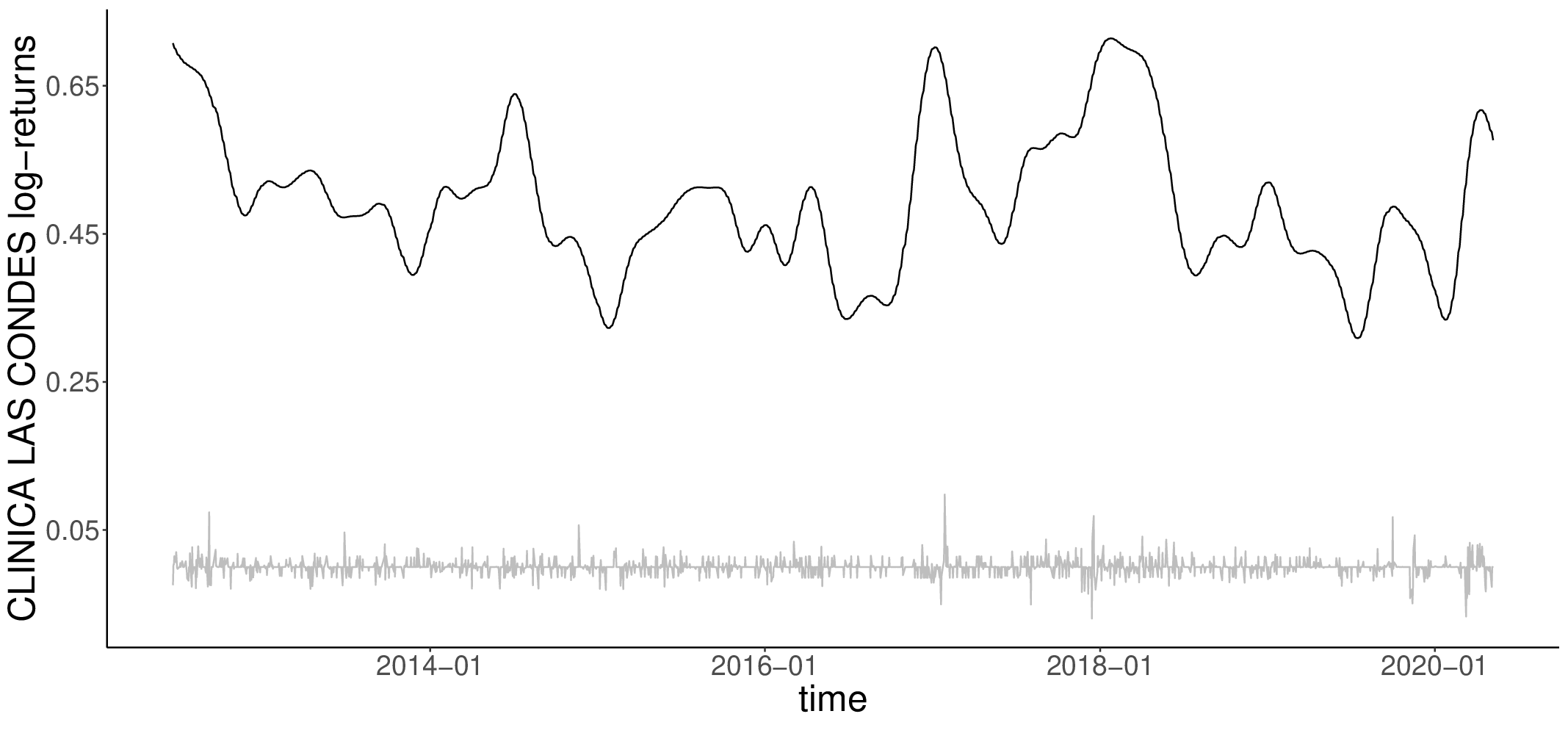}
\protect \includegraphics{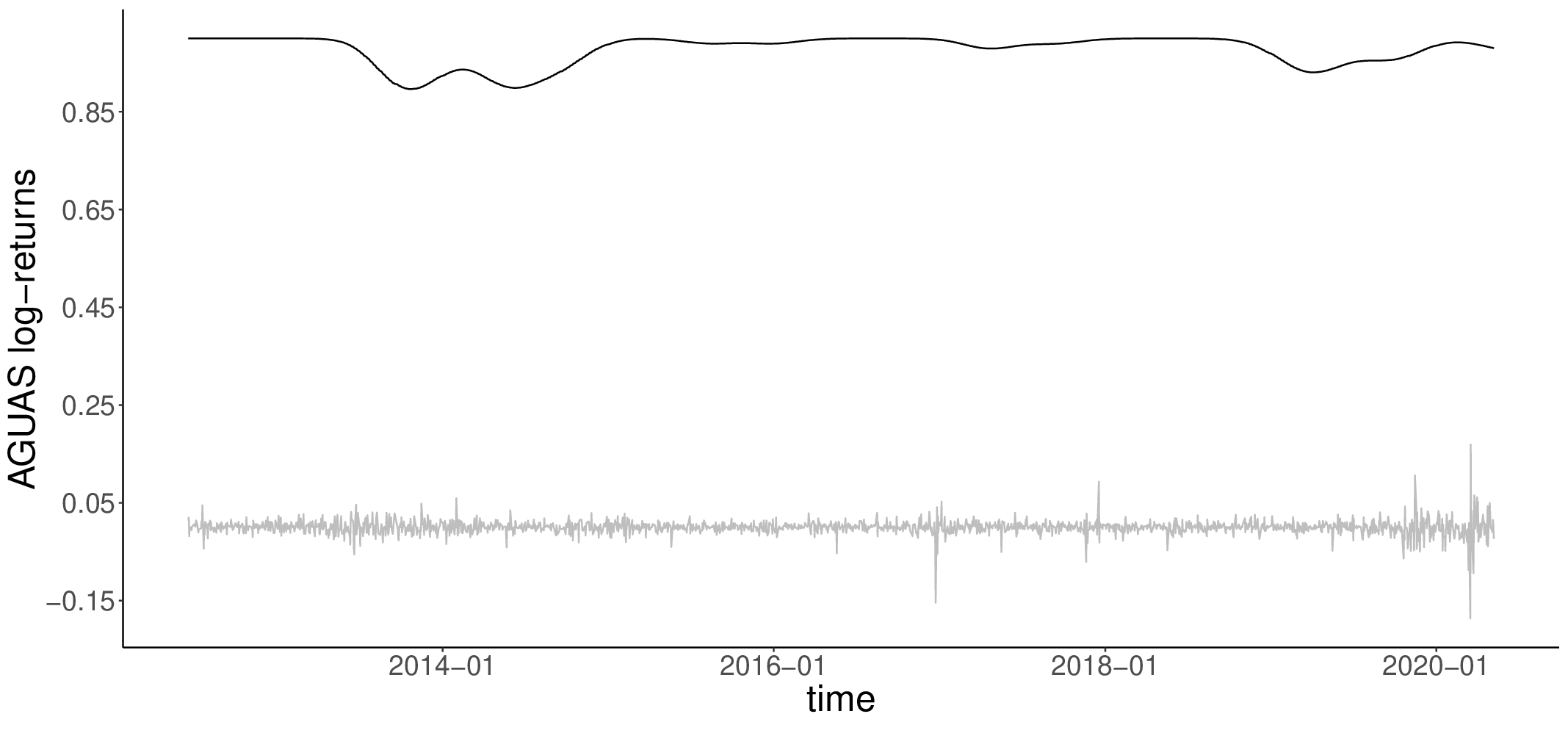}
\caption{\label{s-stocks}
{\footnotesize The same as for Figure \ref{ns-stocks}, but for stocks that seem to have a stationary behavior.}}
\end{figure}

\begin{figure}[h]
\begin{center}
 \includegraphics[scale=0.35]{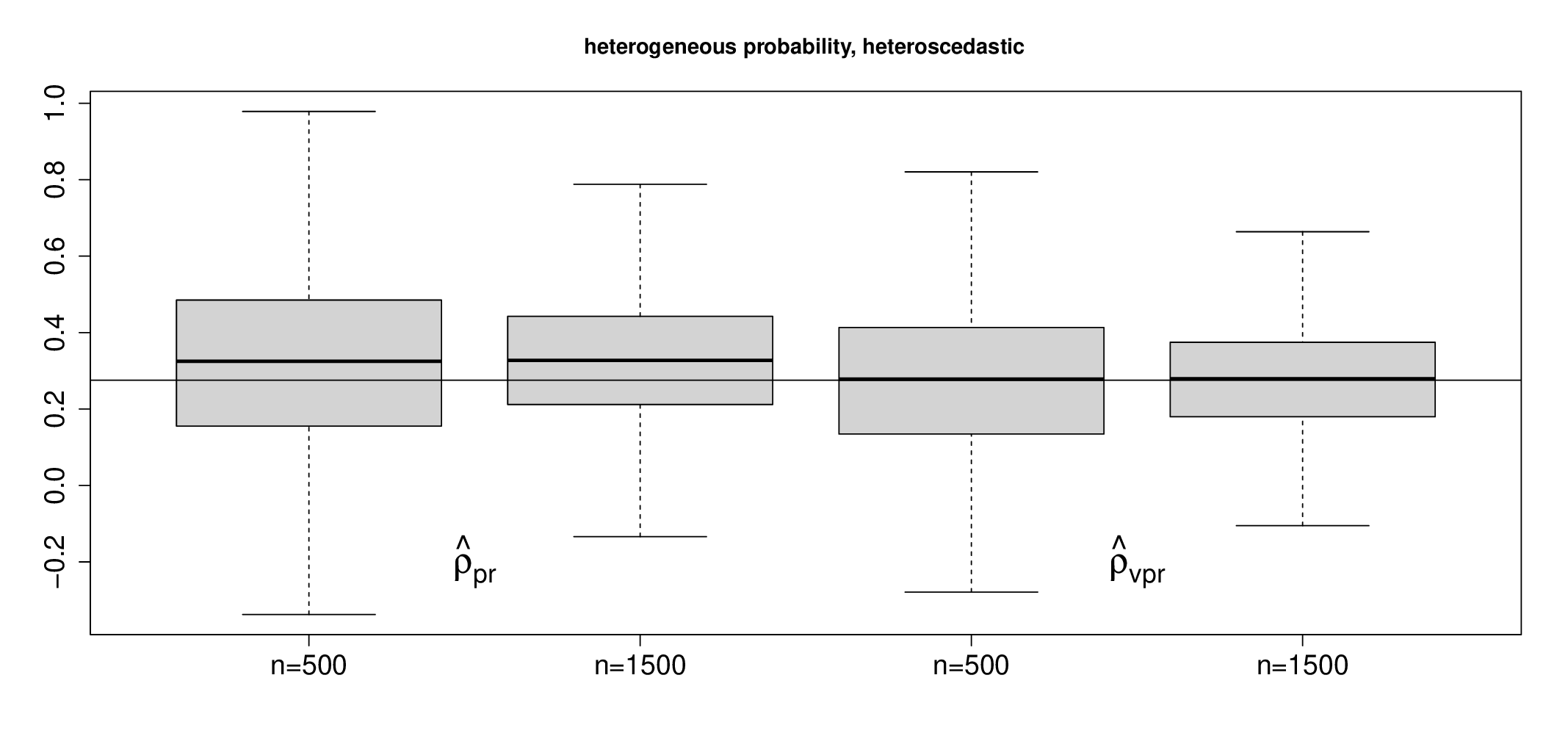}
\end{center}
\caption{The boxplots of the "pr" and "vpr" autocorrelations of order 1, obtained from $N=3000$ independent trajectories of (\ref{dgp}) using the case (i) (i.e. both time-varying $P(a_t=1)$ and $\sigma_t$).}
\label{fig-sim-1}
\end{figure}

\clearpage

\begin{figure}[ht]
  \centering
  \includegraphics[scale=0.35]{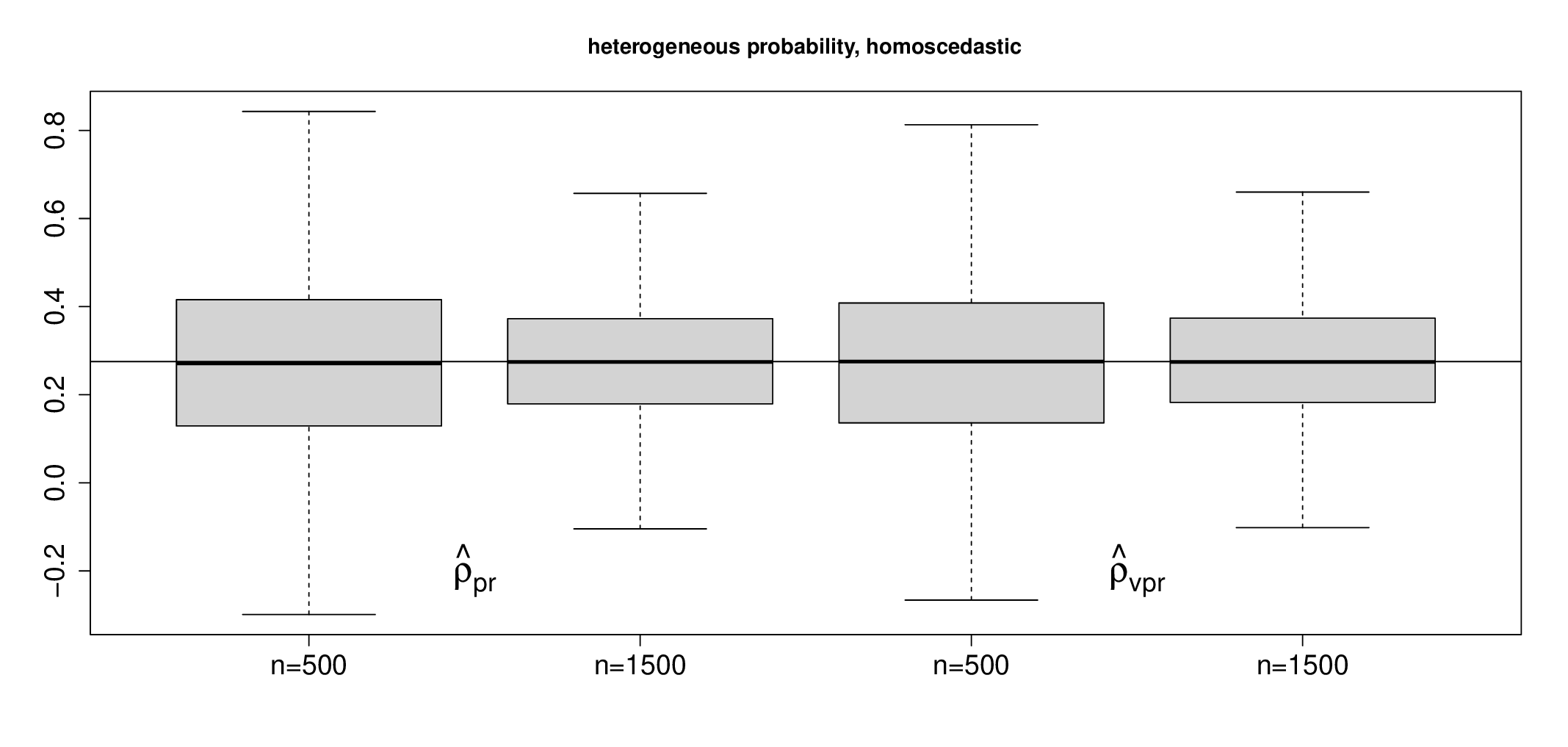}
  \caption{The same as in Figure \ref{fig-sim-1}, but for the case (ii) (i.e. constant $P(a_t=1)$ and time-varying $\sigma_t$).}
\label{fig-sim-2}
  \vspace*{\floatsep}

  \includegraphics[scale=0.35]{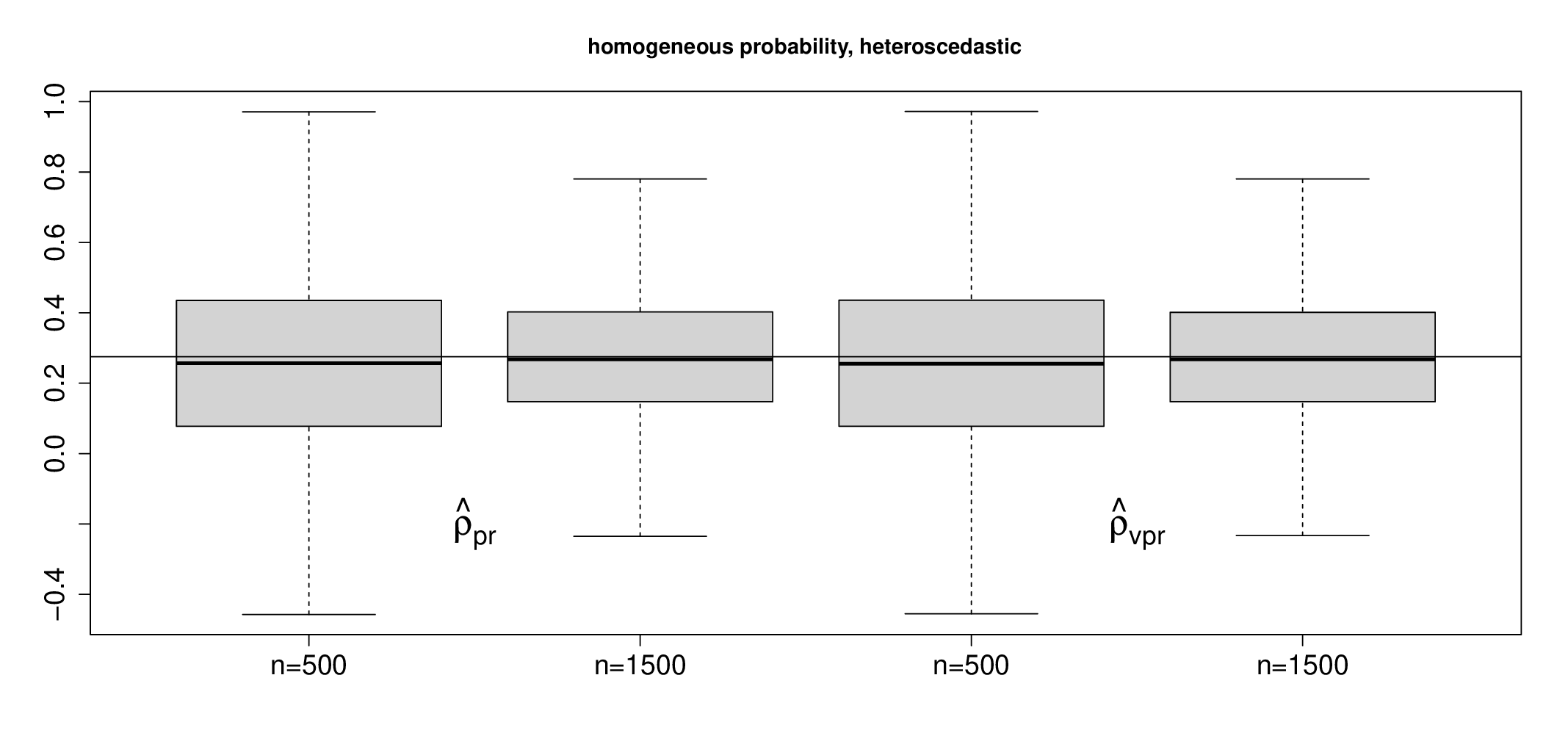}
  \caption{The same as in Figure \ref{fig-sim-1}, but for the case (iii) (i.e. time-varying $P(a_t=1)$ and constant $\sigma_t$).}
\label{fig-sim-2bis}
    \vspace*{\floatsep}

  \includegraphics[scale=0.35]{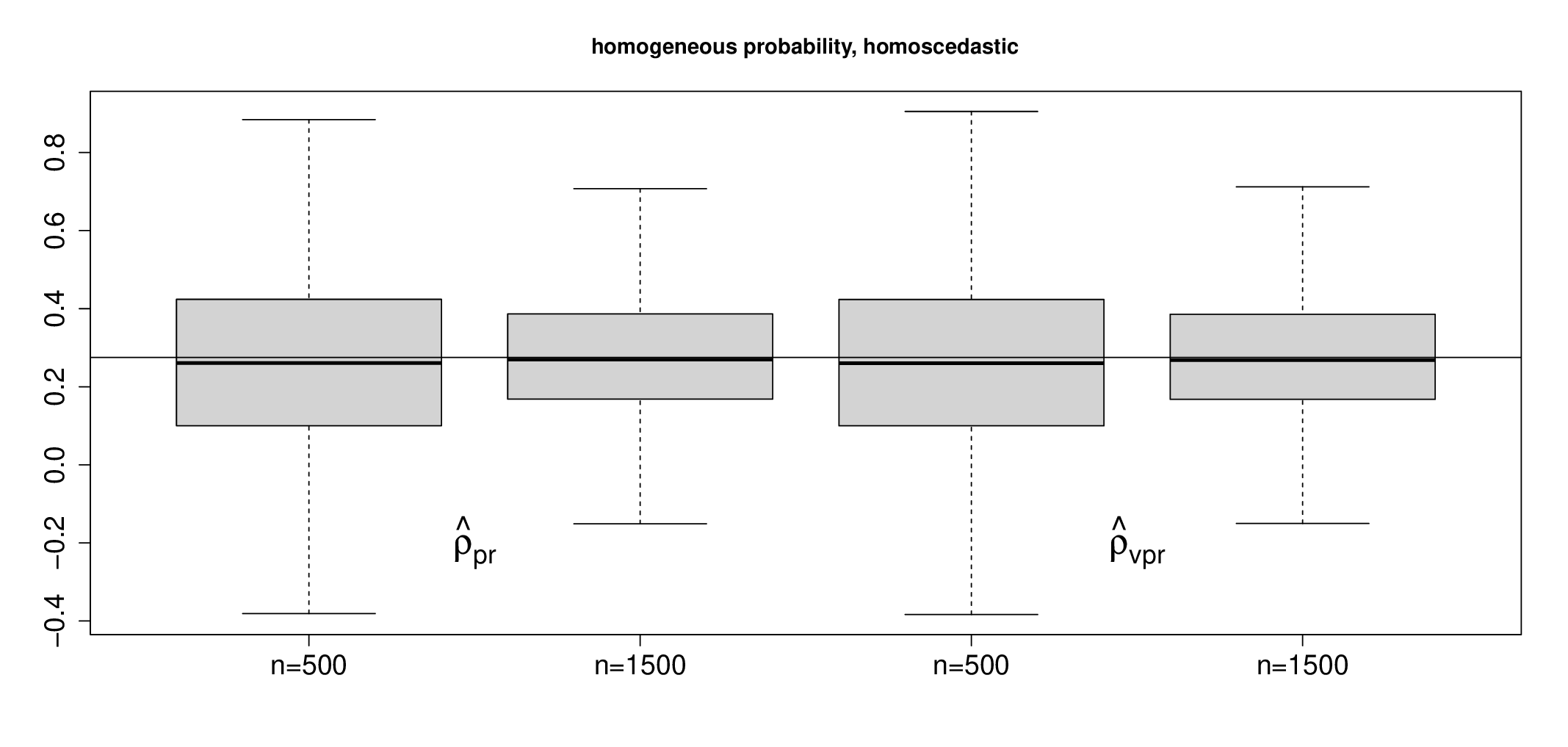}
  \caption{The same as in Figure \ref{fig-sim-1}, but for the case (iv) (i.e. both constant $P(a_t=1)$ and $\sigma_t$).}
\label{fig-sim-2bisbis}
\end{figure}

%
%

\clearpage

\begin{figure}[ht]
  \centering
  \includegraphics[scale=0.35]{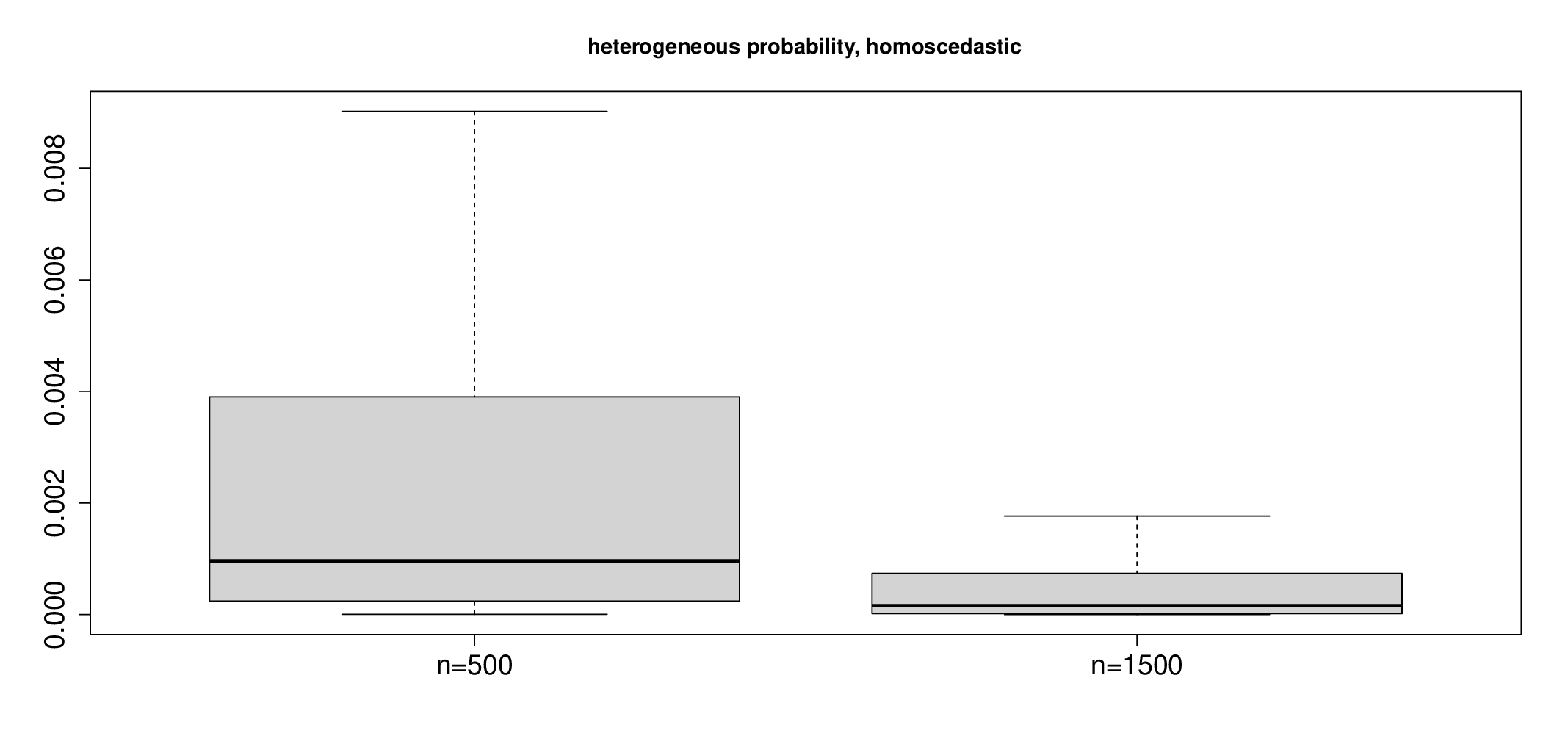}
  \caption{The boxplots of the index $\hat{\kappa}_1$, obtained from $N=3000$ independent trajectories of (\ref{dgp}) considering case (iv) (i.e. both constant $P(a_t=1)$ and $\sigma_t$).}
\label{fig-sim-ind-1}
  \vspace*{\floatsep}

  \includegraphics[scale=0.35]{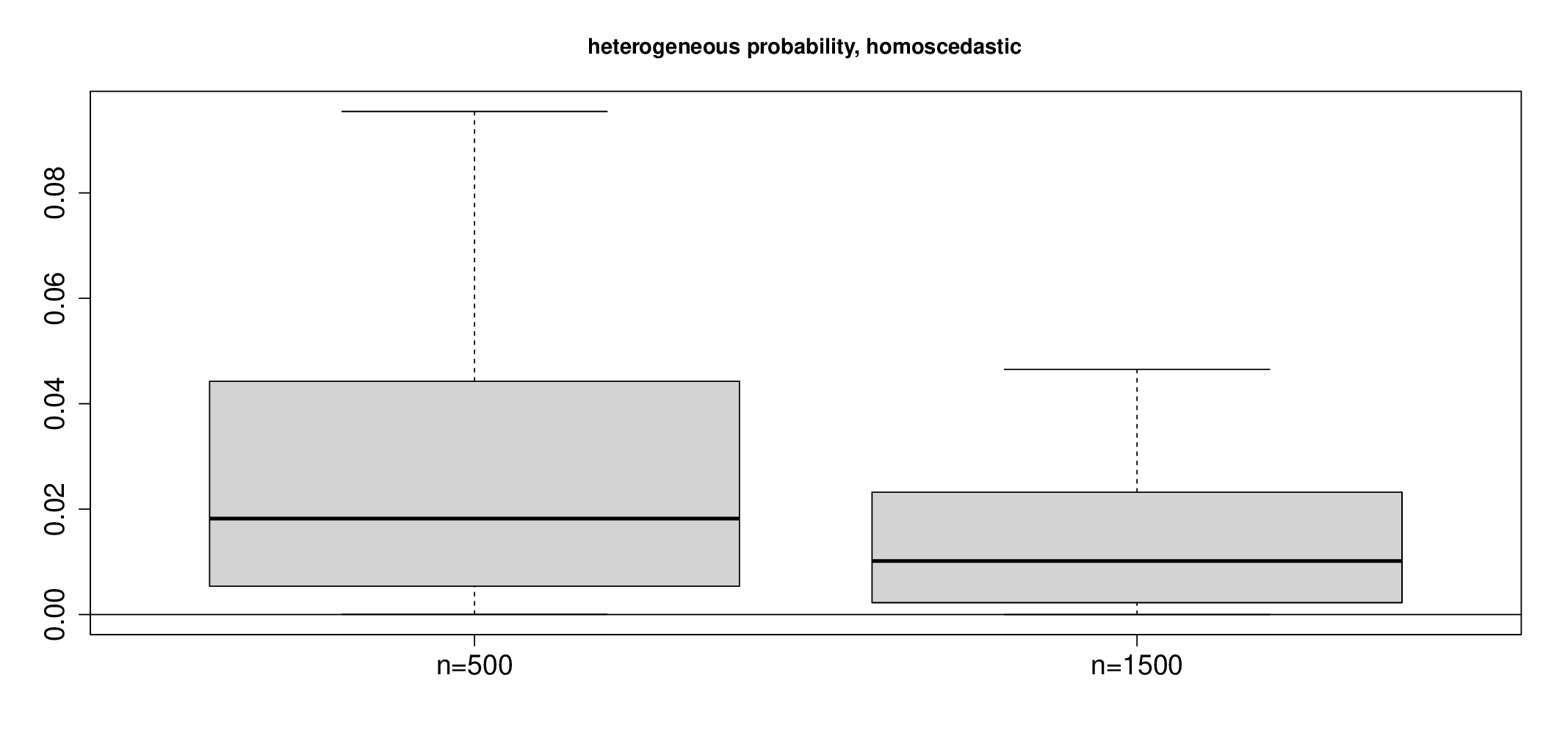}
  \caption{The same as in Figure \ref{fig-sim-ind-1}, but for the case (ii) (i.e. time-varying $P(a_t=1)$ and constant $\sigma_t$).}
\label{fig-sim-ind-1bis}
    \vspace*{\floatsep}

  \includegraphics[scale=0.35]{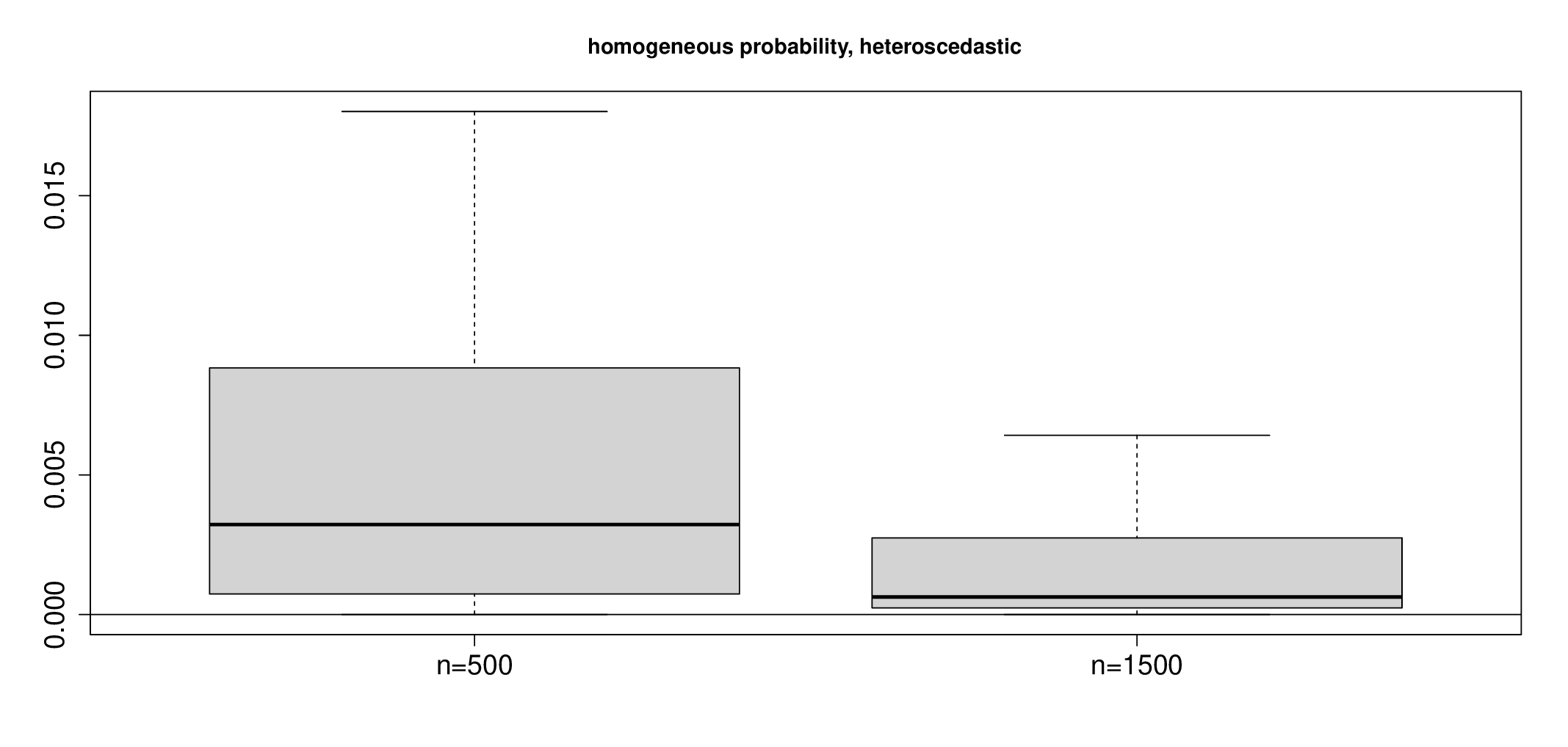}
  \caption{The same as in Figure \ref{fig-sim-ind-1}, but for the case (iii) (i.e. constant $P(a_t=1)$ and time-varying $\sigma_t$).}
\label{fig-sim-ind-1bisbis}
\end{figure}

\clearpage

\begin{figure}[h]
\begin{center}
 \includegraphics[scale=0.35]{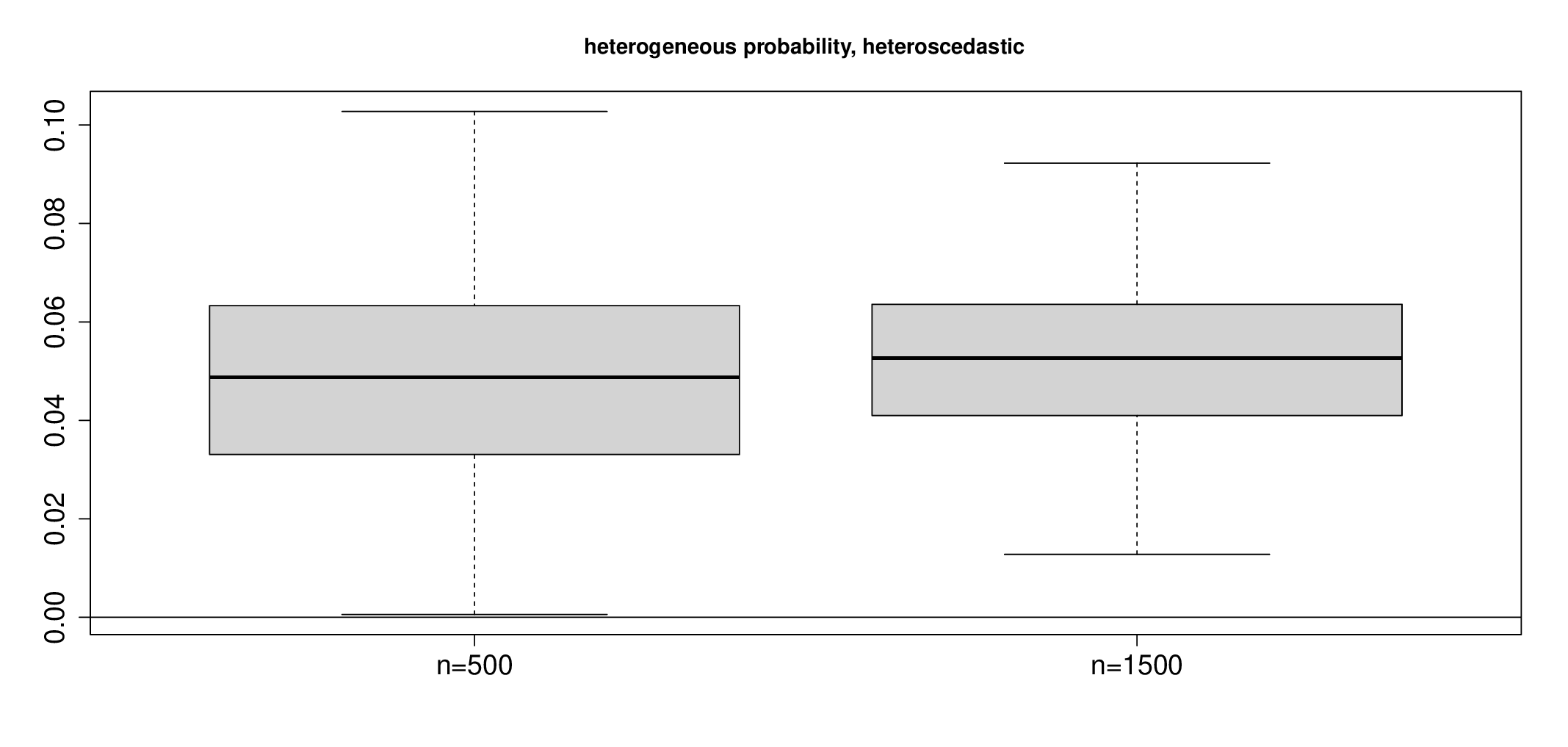}
\end{center}
\caption{The same as in Figure \ref{fig-sim-ind-1}, but for the case (i) (i.e. both time-varying $P(a_t=1)$ and $\sigma_t$).}
\label{fig-sim-ind-2}
\end{figure}

%
%
%

 \end{document}